\documentclass[envcountsame]{llncs}
\usepackage[T1]{fontenc}
\usepackage[utf8]{inputenc}
\usepackage{amssymb,amsmath}
\usepackage{bm,mathrsfs}
\usepackage{xcolor}
\usepackage{graphicx}
\usepackage{url}
\usepackage{array,multirow}
\usepackage{orcidlink}
\usepackage{xifthen}
\usepackage{subcaption}
\usepackage{stmaryrd}
\usepackage{xspace}

\usepackage{booktabs}
\usepackage{numprint}
\npstyleenglish

\usepackage{pgfplots}
\pgfplotsset{
compat=1.17,
mystyle/.style={
    scale only axis,
    width=0.7\columnwidth,
    height=0.5\columnwidth,
    label style={inner sep=0, font=\normalsize},
    tick label style={font=\scriptsize},
    legend style={font=\scriptsize},
    mark size=3,
    major grid style={dashed},
    line width=0.8pt,
    axis line style = thin}
}

\usepackage{tikz}
\usetikzlibrary{calc, trees, positioning, arrows, arrows.meta, chains,%
shapes, shapes.geometric, shapes.symbols, decorations.pathreplacing,%
 decorations.pathmorphing, matrix, quotes, patterns, fit, spy, calligraphy, plotmarks, tikzmark}

\usepackage[ruled, lined,titlenumbered,linesnumbered]{algorithm2e}
\DontPrintSemicolon
\SetKwInOut{Input}{Input}
\SetKwInOut{Output}{Output}
\SetKwComment{Comment}{}{}

\SetKwRepeat{DoWhile}{do}{while}
\SetKwBlock{Repeat}{Repeat}{}
\SetKwProg{Fn}{Function}{:}{}
\SetArgSty{textup}

\usepackage{mathtools}
\mathtoolsset{showonlyrefs}
\allowdisplaybreaks

\usepackage[nolist]{acronym}

\usepackage[capitalize,nameinlink]{cleveref}

\newcommand{\ReSo}{Reed--Solomon\xspace}
\newcommand{\reskew}{ReSkew\xspace}

\newcommand{\defeq}{:=}

\renewcommand{\mod}{\; \textnormal{ mod } \;}

\newcommand{\Fq}{\ensuremath{\mathbb F_{q}}}
\newcommand{\Fqm}{\ensuremath{\mathbb F_{q^m}}}

\newcommand{\Polyring}{\ensuremath{\Fqm[x]}}

\newcommand{\Linpolyring}{\mathbb{L}_{q^m}\![x]}

\newcommand{\SkewPolyringZeroDer}{\ensuremath{\Fqm[x;\aut]}}

\newcommand{\NN}{\ensuremath{\mathbb{N}}}

\newcommand{\RR}{\ensuremath{\mathbb{R}}}

\newcommand{\pe}{\ensuremath{\zeta}}

\newcommand{\aut}{\ensuremath{\theta}}
\newcommand{\autinv}{\ensuremath{\aut^{-1}}}

\newcommand{\frob}{\sigma}

\newcommand{\remev}[2]{{#1}\!\left[#2\right]}
\newcommand{\opev}[3]{\ensuremath{{#1}(#2)_{#3}}}
\newcommand{\opfull}[2]{\ensuremath{\mathcal{D}_{\aut}(#2)_{#1}}}

\newcommand{\opfullexp}[3]{\ensuremath{\mathcal{D}_{\aut}^{#3}(#2)_{#1}}}
\newcommand{\opfullexpinv}[3]{\ensuremath{\mathcal{D}_{\autinv}^{#3}(#2)_{#1}}}

\newcommand{\opexp}[3]{\ensuremath{\mathcal{D}^{#3}(#2)_{#1}}}

\newcommand{\OCompl}[1]{\ensuremath{\mathcal{O}({#1})}}

\DeclareMathOperator{\wt}{wt}
\DeclareMathOperator{\rk}{rk}

\DeclareMathOperator{\diag}{diag}

\newcommand{\mat}[1]{\ensuremath{\bm{#1}}}

\let\Oslash\O %

\renewcommand{\a}{\boldsymbol{a}}
\renewcommand{\b}{\boldsymbol{b}}
\renewcommand{\c}{\boldsymbol{c}}
\renewcommand{\d}{\boldsymbol{d}}
\newcommand{\e}{\boldsymbol{e}}

\newcommand{\m}{\boldsymbol{m}}
\newcommand{\n}{\boldsymbol{n}}

\newcommand{\s}{\boldsymbol{s}}

\renewcommand{\v}{\boldsymbol{v}}
\newcommand{\w}{\boldsymbol{w}}
\newcommand{\x}{\boldsymbol{x}}
\newcommand{\y}{\boldsymbol{y}}

\newcommand{\vecalpha}{\boldsymbol{\alpha}}

\newcommand{\veclambda}{\boldsymbol{\lambda}}

\newcommand{\G}{\boldsymbol{G}}
\renewcommand{\H}{\boldsymbol{H}}
\newcommand{\I}{\boldsymbol{I}}

\newcommand{\M}{\boldsymbol{M}}

\renewcommand{\O}{\boldsymbol{O}}
\renewcommand{\P}{\boldsymbol{P}}

\newcommand{\R}{\boldsymbol{R}}

\newcommand{\T}{\boldsymbol{T}}
\newcommand{\U}{\boldsymbol{U}}

\newcommand{\0}{\ensuremath{\mathbf 0}}
\newcommand{\1}{\ensuremath{\mathbf 1}}

\newcommand{\cZero}{\ensuremath{\c_{0}}}
\newcommand{\cHat}{\ensuremath{\widehat{\c}}}

\newcommand{\mycode}[1]{\ensuremath{\mathcal{#1}}}

\newcommand{\Gab}[1]{\mathrm{Gab}\ensuremath{[#1]}}

\newcommand{\GRS}[1]{\ensuremath{\mathrm{GRS}[#1]}}

\newcommand{\skewRS}[1]{\ensuremath{\mathrm{SRS}[#1]}}
\newcommand{\genSkewRS}[1]{\ensuremath{\mathrm{GSRS}[#1]}}
\newcommand{\genRS}[1]{\ensuremath{\mathrm{GRS}[#1]}}

\newcommand{\linRS}[1]{\ensuremath{\mathrm{LRS}[#1]}} %

\newcommand{\GLRS}[1]{\ensuremath{\mathrm{GLRS}[#1]}}

\newcommand{\shot}[2]{\ensuremath{{#1}^{(#2)}}}
\newcommand{\subShot}[3]{\ensuremath{{#1}^{(#3)}}_{#2}}
\newcommand{\shots}{\ensuremath{\ell}}
\newcommand{\len}{\ensuremath{n}}
\newcommand{\lenShot}[1]{\ensuremath{\len_{#1}}}

\newcommand{\wtH}{\ensuremath{\wt_{H}}}
\newcommand{\wtRk}{\ensuremath{\wt_{R}}}
\newcommand{\wtSrk}{\ensuremath{\wt_{\Sigma R}}}

\newcommand{\wtSkewIdx}[1]{\ensuremath{\wt_{S}^{#1}}}

\newcommand{\dH}{\ensuremath{d_{H}}}
\newcommand{\dRk}{\ensuremath{d_{R}}}
\newcommand{\dSrk}{\ensuremath{d_{\Sigma R}}}

\newcommand{\dSkewIdx}[1]{\ensuremath{d_{S}^{#1}}}

\newcommand{\distH}{d_H}

\newcommand{\Gsec}{\ensuremath{\G_{\text{\normalfont sec}}}}

\newcommand{\Hpub}{\ensuremath{\H_{\text{\normalfont pub}}}}

\newcommand{\skewVandermonde}[2]{\ensuremath{\mat{V}_{\aut}^{#1}(#2)}}
\newcommand{\skewVandermondeInv}[2]{\ensuremath{\mat{V}_{\autinv}^{#1}(#2)}}
\newcommand{\opVandermonde}[3]{\ensuremath{\mathfrak{M}_{\aut}^{#1}(#2)_{#3}}}
\newcommand{\opVandermondeInv}[3]{\ensuremath{\mathfrak{M}_{\autinv}^{#1}(#2)_{#3}}}

\newcommand{\gcrd}{\ensuremath{\mathrm{gcrd}}}
\newcommand{\lclm}{\ensuremath{\mathrm{lclm}}}

\newcommand{\functionTitle}[2]{\textsc{#1}\textup{(#2)}}

\DeclareMathOperator{\sk}{sk}
\DeclareMathOperator{\pk}{pk}
\DeclareMathOperator{\params}{params}

\DeclareMathOperator{\id}{id}

\newcommand{\brackAut}[1]{\ensuremath{[#1]_{\aut}}}
\newcommand{\brackWIndex}[2]{\ensuremath{[#1]_{#2}}}

\newcommand{\doubleBrackAut}[1]{\ensuremath{{\llbracket #1 \rrbracket}_{\aut}}}
\newcommand{\doubleBrackAutinv}[1]{\ensuremath{{\llbracket #1 \rrbracket}_{\autinv}}}
\newcommand{\doubleBrackWIndex}[2]{\ensuremath{{\llbracket #1 \rrbracket}_{#2}}}

\newcommand{\PAut}{{P}_{\aut}}
\newcommand{\PAutinv}{{P}_{\autinv}}
\newcommand{\PIdx}[1]{{P}_{#1}}

\newcommand{\kIdx}[1]{k_{#1}}

\newcommand{\ConjAut}[2]{\gamma_{\aut}(#1, #2)}
\newcommand{\ConjFrob}[2]{\gamma_{\frob}(#1, #2)}
\newcommand{\ConjAutinv}[2]{\gamma_{\autinv}(#1, #2)}

\newcommand{\puncCIdx}[1]{\mycode{C}_{\text{pun}, #1}}
\newcommand{\puncVarIdx}[2]{#1_{\text{pun}, #2}}
\newcommand{\shortCIdx}[1]{\mycode{C}_{\text{sho}, #1}}
\newcommand{\shortVarIdx}[2]{#1_{\text{sho}, #2}}

\newcommand{\qdeg}{\deg_{q}}
\newcommand{\mystack}[2]{\ensuremath{\genfrac{}{}{0pt}{}{#1}{#2}}}
\newcommand{\numtimes}[2]{#1} %

\begin{acronym}
	\acro{KEM}{key-encapsulation mechanism}
	\acro{LRS}{linearized Reed--Solomon}
	\acro{GLRS}{generalized linearized Reed--Solomon}
	\acro{NIST}{National Institute of Standards and Technology}
	\acro{PQC}{post-quantum cryptography}
	\acro{RS}{Reed--Solomon}
	\acro{GRS}{generalized Reed--Solomon}
	\acro{SDP}{syndrome-decoding problem}
	\acro{SRS}{skew Reed--Solomon}
	\acro{GSRS}{generalized skew Reed--Solomon}
	\acro{ISD}{information-set decoding}
	\acro{MDS}{maximum distance separable}
	\acro{MSD}{maximum skew distance}
	\acro{MSRD}{maximum sum-rank distance}
	\acro{MRD}{maximum rank distance}
	\acro{BSI}{Federal Office for Information Security}
	\acro{DSA}{digital signature algorithm}
	\acro{HQC}{Hamming Quasi-Cyclic}
	\acro{BIKE}{Bit Flipping Key Encapsulation}
	\acro{QC-MDPC}{quasi-cyclic moderate-density parity-check}
	\acro{DFR}{decryption failure rate}
	\acro{PKE}{public-key encryption}
	\acro{PKC}{public-key cryptosystem}
	\acro{IND-CPA}{indistinguishability under chosen-plaintext attacks}
	\acro{OW-CPA}{one-way under chosen-plaintext attacks}
	\acro{IND-CCA2}{indistinguishability under adaptive chosen-ciphertext attacks}
	\acro{ROM}{random oracle model}
	\acro{QROM}{quantum random oracle model}
	\acro{RAM}{random access machine}
	\acro{lclm}{least common left multiple}
	\acro{gcrd}{greatest common right divisor}
	\acro{ISO}{International Organization for Standardization}
	\acro{GSE}{generalized skew evaluation}
\end{acronym}

\title{Distinguishers for Skew and Linearized Reed--Solomon Codes}

\author{Felicitas Hörmann\inst{1,2}\,\orcidlink{0000-0003-2217-9753} \and
Anna-Lena Horlemann\inst{2}\,\orcidlink{0000-0003-2685-2343}}

\institute{Institute of Communications and Navigation, German Aerospace Center (DLR), Oberpfaffenhofen--Wessling, Germany \\
\email{felicitas.hoermann@dlr.de}
\and
Institute of Computer Science, University of St.Gallen, St.\ Gallen, Switzerland \\
\email{anna-lena.horlemann@unisg.ch}}

\begin{document}
	
	\maketitle
	
	\begin{abstract}
		Generalized Reed--Solomon (GRS) and Gabidulin codes have been proposed for various code-based cryptosystems, though most such schemes without elaborate disguising techniques have been successfully attacked. Both code classes are prominent examples of the isometric families of (generalized) skew and linearized Reed--Solomon ((G)SRS and (G)LRS) codes which are obtained as evaluation codes from skew polynomials. Both GSRS and GLRS codes share the advantage of achieving the maximum possible error-decoding radius and thus promise smaller key sizes than e.g.\ Classic McEliece.
		
		We investigate whether these generalizations can avoid the known structural attacks on GRS and Gabidulin codes. In particular, we prove that both GSRS and GLRS codes decompose into GRS subcodes and are thus efficiently distinguishable from random codes with a square code method.
		This applies to all parameters for which the code length $n$ and its dimension $k$ over the field $\Fqm$ satisfy $m + 1 < k < n - \tfrac{1}{2} (m^2 + 3m)$. The distinguishability extends to GSRS and GLRS codes with Hamming-isometric disguising.
		
		We further relate these findings to existing distinguishers for GRS, Ga\-bi\-du\-lin, and LRS codes, and extend known results on duals of SRS and LRS codes to the generalized setting allowing nonzero column multipliers. Finally, we provide explicit transformations between GSRS and GLRS codes, clarifying the algebraic relationship between the skew and linearized frameworks.
		
	\end{abstract}

	\acresetall

	\section{Introduction}
	
	\paragraph{Post-quantum cryptography.}
	We have known for decades that quantum computers pose a threat for current public key cryptography based on integer factorization or variants of the discrete logarithm problem.
	This awareness lead to the emergence of the field of \ac{PQC}, which studies alternative security assumptions and develops cryptographic algorithms that resist both classical and quantum attacks.
	Since there has been tremendous progress in quantum computing and store-now-decrypt-later attacks constitute an additional risk, the transition to \ac{PQC} is not only a crucial but also a time-pressing matter.
	For instance, Germany's \ac{BSI} conservatively estimates in its 2024 status report on quantum computing that cryptographically relevant quantum computers will become available by around 2040~\cite{BSIquantum}.

	The U.S.\ \ac{NIST} started a \ac{PQC} standardization project by publishing a call for quantum-resistant algorithms in late 2016.
	While the branch for standardizing \acp{DSA} is still ongoing, the part dealing with \acp{KEM} recently came to an end.
	There, \ac{NIST} selected the lattice-based CRYSTALS-Kyber after three rounds in 2022 and the code-based {HQC} after the fourth round in 2025~\cite{AlagicAponEtAl2022StatusReportFourth}.
	
	\paragraph{Code-based cryptography.}
	McEliece pioneered the field of code-based cryptography with his seminal paper~\cite{McEliece1978PublicKeyCryptosystem} in 1978.
	There, Alice picks a linear code with an efficient decoding algorithm and publishes a random-looking generator matrix of an equivalent code as public key.
	Bob then encrypts a message by encoding it and adding a random correctable error.
	Since Alice knows an efficient decoder for the secret code, she can decrypt successfully, while the adversary Eve faces the hard problem of decoding in a seemingly random code.
	Note that Eve's problem is in fact NP-complete if the input code is arbitrary and gets decoded by means of a maximum-likelihood decoder~\cite{BerlekampMcElieceEtAl1978InherentIntractabilityCertain,Barg1994SomeNewNp}.
	Niederreiter followed a somewhat dual approach when he proposed another cryptosystem based on algebraic codes in~\cite{Niederreiter1986KnapsackTypeCryptosystems}.
	In his framework, the public key is a random-looking parity-check matrix of the secret code and the ciphertexts are syndromes.
	Both formulations were shown to be equivalent from a security point of view~\cite{LiDengEtAl1994EquivalenceMceliecesNiederreiters} and we usually include both variants when referring to McEliece-like schemes for simplicity.
	
	The major shortcoming of McEliece-like systems are their large keys, and the literature recorded numerous suggestions to improve upon this in the last decades.
	Most proposals use strategies like more advanced disguising techniques or alternative decoding metrics, or they employ other code families.
	While the original McEliece system guarantees correct decryption,
	proposals such as {HQC} and \acs{BIKE} use decoders with decoding failures.
	The resulting decryption failures require a careful analysis and potentially lead to side-channel vulnerabilities.
	Unfortunately, most of the systems employing algebraic codes with non-failing decoders and isometric disguising were broken by structural attacks which allow to recover the secret code structure from the public key.

	\paragraph{Cryptosystems using generalized Reed--Solomon and Gabidulin codes.}
	Prominent choices for algebraic code families are \ac{GRS} codes in the Hamming metric and their rank-metric analogs, that is, Gabidulin codes.
	Both classes of codes exhibit the highest possible unique decoding radius and have efficient decoders.
	As a consequence, the public keys of respective cryptosystems are expected to be small due to the rising cost of \ac{ISD} for a growing error weight.
	Another argument for \ac{GRS} and Gabidulin codes is their compact algebraic representation as evaluation codes which yields small secret keys.
	In the following, we focus on schemes based on the unique decoding of \ac{GRS} and Gabidulin codes, and exclude proposals using e.g.\ interleaved codes, subcode constructions, or list-decoding problems.
	
	\paragraph{Structural attacks on GRS codes.}
	Niederreiter developed the first McEliece-like scheme based on \ac{GRS} codes~\cite{Niederreiter1986KnapsackTypeCryptosystems} but Sidelnikov and Shestakov showed that the secret code parameters can be recovered from the isometrically disguised generator matrix, i.e., from the public key~\cite{SidelnikovShestakov1992InsecurityCryptosystemsBased}.
	Wieschebrink introduced additional random columns in the public scrambled generator matrix to mitigate this attack~\cite{Wieschebrink2006TwoNpComplete}, yet the system was broken by the square code distinguisher~\cite{Wieschebrink2010CryptanalysisNiederreiterPublic,CouvreurGaboritEtAl2014DistinguisherBasedAttacks}.
    Indeed, the distinguisher allows to extract the \ac{GRS} columns from the public key and thus makes the scheme susceptible to Sidelnikov--Shestakov's attack.
	The BBCRS scheme~\cite{BaldiBianchiEtAl2014EnhancedPublicKey} mitigates structural attacks with a special disguising technique which amplifies the Hamming weight of the error for the attacker but keeps it decodable for the legitimate receiver.
	Its updated parameter sets appear to be still secure after the attack from~\cite{CouvreurOtmaniEtAl2015PolynomialTimeAttack} broke all initial parameter suggestions.
	
	\paragraph{Structural attacks on Gabidulin codes.}
	The GPT system~\cite{GabidulinParamonovEtAl1991IdealsOverNon} was the start of a series of cryptosystems based on Gabidulin codes~\cite{Gabidulin2008AttacksCounterAttacks,RashwanGabidulinEtAl2010SmartApproachGpt}.
	Unfortunately, the attacking side was strong and many of these proposals were broken by Gibsons's attacks~\cite{Gibson1995SeverelyDentingGabidulin,Gibson1996SecurityGabidulinPublic} as well as by Overbeck's attack~\cite{Overbeck2005NewStructuralAttack,Overbeck2007StructuralAttacksPublic,Overbeck2007PublicKeyCryptography} and its extensions~\cite{Horlemann-TrautmannMarshallEtAl2016ConsiderationsRankBased,Horlemann-TrautmannMarshallEtAl2018ExtensionOverbecksAttack,OtmaniKalachiEtAl2017ImprovedCryptanalysisRank}.
	In contrast to these schemes, Loidreau's cryptosystem~\cite{Loidreau2016EvolutionGptCryptosystem,Loidreau2017NewRankMetric} employs Gabidulin codes but hides their structure by means of rank amplifiers.
	Some of the originally chosen parameter sets were broken by the Coggia--Couvreur attack~\cite{CoggiaCouvreur2020SecurityLoidreauRank} but new parameter choices fixed the vulnerability.
	Recently, a generalization of the system was presented in~\cite{NouetowaLoidreau2025AnalysisGeneralizationLoidreaus}.
	
	\paragraph{Skew evaluation codes.}
	\ac{GRS} and Gabidulin codes are special cases of the isometric families of \acf{GSRS} and \acf{GLRS} codes (see \cite{Martinez-Penas2018SkewLinearizedReed} for example), which are both obtained as evaluation codes from skew polynomial rings.
	The efficient decoding of these codes was showcased by means of various decoding techniques in e.g. \cite{BoucherUlmer2014LinearCodesUsing,LiuManganielloEtAl2015ConstructionDecodingGeneralized,Martinez-PenasKschischang2019ReliableSecureMultishot}.
	\ac{GSRS} and \ac{GLRS} codes are \ac{MDS} and thus promising for reducing the key size in cryptographic applications while potentially being resistant against the known structural attacks against \ac{GRS} and Gabidulin codes.
	Within this research direction, e.g. the distinguishing of \ac{GLRS} codes under isometric disguising was studied in \cite{HoermannBartzEtAl2023DistinguishingRecoveringGeneralized}.
	The main result was an Overbeck-like distinguisher whose complexity is polynomial-time only for trivial column multipliers and partial knowledge of the secret code parameters.
	Further, a cryptosystem based on \ac{GLRS} codes with trivial column multipliers was introduced recently \cite{Nouetowa2026PostQuantumEncryption}.
	The proposal does not use isometric disguising but hides the secret code structure by means of low-dimensional subspaces similar to Loidreau's idea for Gabidulin codes.
	
	\paragraph{Contributions.}
	We study the algebraic structure and the distinguishability of skew evaluation codes with the goal of understanding if they are promising candidates for McEliece-like cryptosystems with isometric disguising.
	In particular, we show that \acf{GSRS} codes and \acf{GLRS} codes can be expressed as the direct sum of \ac{GRS} codes and thus have a predictable square code dimension.
	This structural perspective yields a polynomial-time distinguisher for \ac{GSRS} and \ac{GLRS} codes over $\Fqm$, as long as their dimension $k$ and length $n$ satisfy $m + 1 < k < n - \tfrac{1}{2} (m^2 + 3m)$.
	The results also carry over to codes disguised by Hamming-metric isometries.
	
	We try to give a comprehensive overview of the connections between the isometric \ac{GSRS} and \ac{GLRS} codes as well as the relations to their special cases of \ac{GRS} and Gabidulin codes.
	Thus, we compare our results to known distinguishers and also enhance known results on the duals of \ac{GSRS} and \ac{GLRS} codes for trivial column multipliers to the general setting.
	
	\paragraph{Outline.}
	After this introduction, we start the paper by collecting useful preliminaries about codes and skew polynomials in~\cref{sec:prelims}.
	\cref{sec:skew_eval_codes} introduces \ac{GSRS} and \ac{GLRS} codes, gives results about their duals, and highlights \ac{GRS} and Gabidulin codes as special cases.
	The end of the section explains the connection between the distinguishing problem and security in code-based cryptography.
	In \cref{sec:grs_subcodes}, we show that \ac{GSRS} and \ac{GLRS} codes can be decomposed into \ac{GRS} subcodes.
	This implies the existence of a square code distinguisher which we present in \cref{sec:square_code_dist}.
	We discuss connections to known distinguishers for \ac{GRS}, Gabidulin, and \ac{LRS} codes in \cref{sec:other_distinguishers} and showcase experimental results in \cref{sec:experiments}.
	Finally, we give a summary and mention interesting open research questions in~\cref{sec:conclusion}.
	
	In addition to the main part of the paper, \cref{sec:reskew} introduces the cryptosystem \reskew based on \textbf{Re}ed--\textbf{S}olomon codes in a s\textbf{kew} setting.
	It is a straightforward instantiation of the McEliece system with respect to the Hamming metric and isometric disguising and thus showcases a system for which the presented distinguisher is relevant.
	Further, it quantifies the potential of \ac{GSRS} codes in a McEliece-like setting by providing key and ciphertext sizes for parameter sets adhering to the \ac{NIST} security levels.
	In fact, the size of the public key is reduced by a factor of at least three compared to Classic McEliece.

	\section{Preliminaries}
	\label{sec:prelims}
	
	We mainly work in finite extension fields and denote by $\Fq$ and $\Fqm$ the fields of prime-power order $q$ and $q^{m}$, respectively.
	Field elements are usually denoted by lowercase letters, and bold lowercase and uppercase letters represent vectors and matrices, respectively.
	Further, $\0$ denotes the all-zero vector or the all-zero matrix, and $\1$ represents the all-one vector or the all-one matrix, both of suitable dimensions depending on the context.
	
	We interpret vectors as row vectors and write e.g.\ $\v = (v_1, \dots, v_n) \in \Fqm^n$.
	If we want to highlight that $\v$ is divided into $\shots$ blocks whose lengths are the entries of the vector $\n = (\lenShot{1}, \dots, \lenShot{\shots}) \in \NN^{\shots}$, we denote it as
	\begin{equation}
	    \v = (\shot{\v}{1} \mid \dots \mid \shot{\v}{\shots})\in \Fqm^{\n}
	    \text{ with } \shot{\v}{i} = \bigl(\subShot{v}{1}{i}, \dots, \subShot{v}{\lenShot{i}}{i}\bigr) \in \Fqm^{\lenShot{i}}
		\text{ for } i = 1, \dots, \shots.
	\end{equation}
	This notation carries over to matrices $\M \in \Fqm^{k \times \n}$.
	For two vectors $\v, \w \in \Fqm^n$, we define their \emph{star product}, which is also known as the \emph{Schur product}, as
	\begin{equation}
	    \v\star \w:= (v_1 w_1,\dots, v_n w_n) \in \Fqm^n.
	\end{equation}
	
	Moreover, we use the shorthand $\Fqm^{\ast} = \Fqm \setminus \{0\}$ to denote the multiplicative group of $\Fqm$.
	Recall that $\Fqm^{\ast}$ is cyclic and $\pe \in \Fqm^{\ast}$ is called a primitive element if it generates $\Fqm^{\ast}$.

	\subsection{Skew polynomials and their evaluation}
	
	A \emph{skew polynomial ring} $\SkewPolyringZeroDer$ contains all formal polynomials with finitely many nonzero coefficients from $\Fqm$.
	The addition in $\SkewPolyringZeroDer$ is the same as for the conventional polynomial ring $\Fqm[x]$.
	However, the multiplication in $\SkewPolyringZeroDer$ is determined by the rule $x a = \aut(a) x$ which allows to swap the indeterminate $x$ with a field element $a \in \Fqm$~\cite{Ore1933TheoryNonCommutative}.
	Here, $\aut: \Fqm \to \Fqm$ denotes a field automorphism of $\Fqm$, i.e., it satisfies the properties $\aut(a + b) = \aut(a) + \aut(b)$ and $\aut(a \cdot b) = \aut(a) \cdot \aut(b)$ for any $a, b \in \Fqm$.
	
	Note that the notion of degree from $\Fqm[x]$ carries over naturally to $\SkewPolyringZeroDer$ and is denoted by $\deg(\cdot)$.
	Further, skew polynomial rings are left and right Euclidean domains, which means that the \emph{\acl{lclm}} and the \emph{\acl{gcrd}} of two skew polynomials $f, g \in \SkewPolyringZeroDer$ exist and are denoted by $\lclm(f,g)$ and $\gcrd(f,g)$, respectively.
	
	\begin{remark}
		\label{rem:general_skew_polys}
		Note that \Oslash{}re \cite{Ore1933TheoryNonCommutative} defined skew polynomials for potentially infinite division rings instead of finite fields and also included a $\aut$-derivation in addition to the automorphism $\aut$.
		However, the finite setting is more appropriate for coding theory and finite division rings are fields due to Wedderburn's theorem \cite{wedderburn1905theorem}.
		In this case, nonzero derivations do not really change the setup since skew polynomial rings with nonzero derivation are isomorphic to the one with zero derivation as detailed in \cite[Prop.~40]{Martinez-Penas2018SkewLinearizedReed}.
		This means that our setting is barely a restriction from a coding-theoretic perspective.
	\end{remark}
	
	\paragraph{Automorphism choices.}
	Each automorphism $\aut: \Fqm \to \Fqm$ has a fixed field and we choose $\aut$ such that this is precisely $\Fq$.
	In other words, $\aut$ has the form $\aut(x) = x^{q^{s}}$ for an $s$ with $\gcd(s, m) = 1$.
	In particular, $\aut$ is the $s$-th power of the \emph{Frobenius automorphism} $\sigma$ which is the $q$-powering map $x \mapsto x^q$ for all $x \in \Fqm$.
	Similar to \cite{Liu2016GeneralizedSkewReed}, we introduce the notations
	\begin{equation}
		\brackAut{i} %
		= q^{is}
		\qquad \text{and} \qquad
		\doubleBrackAut{i} %
		= \sum_{j = 0}^{i-1} q^{js} = q^{(i-1) s} + \dots + q^s + 1
	\end{equation}
	for $\aut = \frob^s$ with an arbitrary $s = 0, \dots, m - 1$ and for all $i \in \NN$.
	In particular, $\aut^i(x) = x^{\brackAut{i}}$ applies and $\doubleBrackAut{i} = \tfrac{q^{is} - 1}{q^{s} - 1}$ for $s \neq 0$.
	
	\paragraph{Remainder evaluation.}
	One approach to skew polynomial evaluation is enforcing a remainder theorem~\cite{Lam1986GeneralTheoryVandermonde,LamLeroy1988VandermondeWronskianMatrices}.
	Namely, the evaluation of $f \in \SkewPolyringZeroDer$ at a point $a \in \Fqm$ is defined as the remainder of the right division of $f$ by $x - a$.
	We will denote it by $\remev{f}{a}$ to visually distinguish it from other function evaluations.
	In fact, {remainder evaluation} can be expressed similarly to usual polynomial evaluation, namely~\cite[Lem.~2.4]{LamLeroy1988VandermondeWronskianMatrices} yields
	\begin{equation}
	    \remev{f}{a} = \sum\nolimits_{i} f_i a^{\doubleBrackAut{i}}
	\end{equation}
	for any $a \in \Fqm$ and $f \in \SkewPolyringZeroDer$ having the form $f(x) = \sum_{i} f_i x^{i}$.
	For vectors $\a \in \Fqm^n$, we write $\remev{f}{\a} = (\remev{f}{a_1}, \dots, \remev{f}{a_n}) \in \Fqm^{n}$.
	
	\paragraph{Generalized operator evaluation.}
	Another way of evaluating skew polynomials is based on the operator $\opfull{a}{b} \defeq \aut(b)a$ for all $a,b \in \Fqm$ and its powers $\opfullexp{a}{b}{i} = \aut^i(b) \cdot a^{\doubleBrackAut{i}}$ for all $i>0$ (see e.g.\ \cite{leroy1995pseudolinear,Martinez-Penas2018SkewLinearizedReed}).
	The \emph{generalized operator evaluation} of $f(x) = \sum_i f_i x^{i} \in \SkewPolyringZeroDer$ at $b \in \Fqm$ with respect to the evaluation parameter $a \in \Fqm$ is then
	\begin{equation}\label{eq:def_gen_op_eval}
	  \opev{f}{b}{a} \defeq \sum\nolimits_{i}f_i\opfullexp{a}{b}{i} = \sum\nolimits_i f_i \aut^{i}(b) a^{\doubleBrackAut{i}}.
	\end{equation}
	Observe that there is the useful relation
	\begin{equation}\label{eq:connection_remainder_gen-op}
		\opev{f}{b}{a} b^{-1} = \remev{f}{\opfull{a}{b} b^{-1}} \quad \text{for any } a, b \in \Fqm
	\end{equation}
	between remainder and generalized operator evaluation \cite[Lem.~24]{Martinez-Penas2018SkewLinearizedReed}.
	For simplicity, we generalize the notations $\opfullexp{\cdot}{\cdot}{i}$ and $\opev{f}{\cdot}{\cdot}$ to vectors $\a, \b \in \Fqm^n$ via elementwise application and understand $\opfullexp{\a}{\b}{i} = (\opfullexp{a_1}{b_1}{i}, \dots, \opfullexp{a_n}{b_n}{i}) \in \Fqm^{n}$ and $\opev{f}{\b}{\a} = (\opev{f}{b_1}{a_1}, \dots, \opev{f}{b_n}{a_n}) \in \Fqm^n$, respectively.
	
	\paragraph{Conjugacy and norm.}

		For any two elements $a\in\Fqm$ and $c\in\Fqm^*$, we define
		\begin{equation}
			\ConjAut{a}{c} \defeq \aut(c)ac^{-1}.
		\end{equation}
		Two elements $a,b\in\Fqm$ are called \emph{$\aut$-conjugates}, if there exists an element $c\in\Fqm^*$ such that $b = \ConjAut{a}{c}$ holds~\cite{Lam1986GeneralTheoryVandermonde}.
	The notion of $\aut$-conjugacy defines an equivalence relation on $\Fqm$ and thus a partition of $\Fqm$ into $\aut$-conjugacy classes~\cite{LamLeroy1988VandermondeWronskianMatrices}.
	We use the shorthand $\ConjAut{\a}{\c}$ for two vectors $\a, \c \in \Fqm^n$ with $\c$ having nonzero entries to denote $(\ConjAut{a_1}{c_1}, \dots, \ConjAut{a_n}{c_n}) \in \Fqm^n$.
	
	The \emph{norm} of $b \in \Fqm$ with respect to $\aut$ is $N_{\aut}(b) = \prod_{j=0}^{m-1} \aut^j(b) = b^{\doubleBrackAut{m}}$.
	It follows from Hilbert's Theorem 90 that two field elements have the same norm if and only if they belong to the same $\aut$-conjugacy class~\cite[Eq.~(1.2)]{LamLeroy1994Hilbert90Theorems}.
	We extend the notation $N_{\aut}(\cdot)$ elementwise to vectors and write $N_{\aut}(\b) = (N_{\aut}(b_1), \dots, N_{\aut}(b_n))$ for any $\b \in \Fqm^n$.
	
	\paragraph{P-independence.}
	We call a vector $\vecalpha \in \Fqm^{n}$ \emph{$\PAut$-independent} if $\lclm_{i=1, \dots, n} (x - \alpha_i) \in \SkewPolyringZeroDer$ has degree $n$ \cite{BoucherNouetowa2025DecodingAlgorithmSkew}.
	Interestingly, P-independence can be expressed in terms of $\Fq$-linear independence when $\vecalpha$ is represented in a particular form.
	Namely, we group the entries by $\aut$-conjugacy class and express each of them as a conjugate of a selected class representative.
	When $a_1, \dots, a_{\shots} \in \Fqm$ belong to distinct conjugacy classes and $\n \in \NN^{\shots}$ is a block partition, let us define the vector $\a \in \Fqm^{\n}$ such that the $i$-th block $\shot{\a}{i} \in \Fqm^{\lenShot{i}}$ of $\a$ contains precisely $\lenShot{i}$ copies of $a_i$.
	This will allow us to keep the notation as simple as possible and we refer to $\a$ as the $\n$-vector corresponding to $a_1, \dots, a_{\shots}$ in the following.
	
	\begin{lemma}
		\label{lem:p-indep_conditions}
		Consider an arbitrary vector $\vecalpha \in \Fqm^n$ whose entries belong to $\shots$ distinct $\aut$-conjugacy classes.
		Pick representatives $a_1, \dots, a_{\shots} \in \Fqm$ of the respective conjugacy classes and denote the corresponding $\n$-vector by $\a \in \Fqm^{\n}$.
		Then there exist a permutation $\pi \in S_n$ and a vector $\b \in \Fqm^{\n}$ with nonzero entries such that $\pi(\vecalpha) = \ConjAut{\a}{\b}$ holds.
	    Further, the following statements are equivalent:
	    \begin{enumerate}
		    \item $\vecalpha \in \Fqm^{n}$ is $\PAut$-independent.
		    \item $\a$ contains only nonzero entries and each block of $\b$ contains $\Fq$-linearly independent elements.
	    \end{enumerate}
	\end{lemma}

	\begin{proof}
		Recall that $\aut$-conjugacy is an equivalence relation and each element $\alpha \in \Fqm$ can thus be represented as a conjugate $\ConjAut{a}{b}$ of a class representative $a \in \Fqm$ with a nonzero $b \in \Fqm$.
		When we fix a set of representatives of the conjugacy classes and let $\pi$ group the entries of $\vecalpha$ by their class membership, we arrive at the claimed representation $\pi(\vecalpha) = \ConjAut{\a}{\b}$.
	    The equivalence of 1. and 2. follows from \cite[Lem.~1]{LiuManganielloEtAl2017MatroidalStructureSkew} and~\cite[Lem.~29]{Martinez-Penas2018SkewLinearizedReed}.
		\qed
	\end{proof}
	
	\begin{remark}
	    \label{rem:sum-rank_blocks}
		We only introduced the permutation $\pi$ in \cref{lem:p-indep_conditions} to simplify the notation.
	    $\pi$ rearranges the entries of $\vecalpha$ by $\aut$-conjugacy class and allows to reflect this structure also in the block partition $\n$ of the vector $\b \in \Fqm^{\n}$.
	    This will be beneficial later in the context of \ac{GLRS} codes and the related sum-rank weight that depends on the block partition $\n$.
	    However, there is no mathematical problem if the entries of the representation $\ConjAut{\a}{\b}$ are not sorted by class.
	\end{remark}

	\subsection{Linear codes}
	
	An \emph{$\Fqm$-linear code} $\mycode{C}$ of length $n$ and dimension $k$ is a $k$-dimensional $\Fqm$-linear subspace of $\Fqm^{n}$.
	We call the ratio $\tfrac{k}{n}$ its \emph{rate}.
	$\mycode{C}$ can be represented as the $\Fqm$-linear row space of a full-rank \emph{generator matrix} $\G \in \Fqm^{k \times n}$ or as the kernel of a full-rank \emph{parity-check matrix} $\H \in \Fqm^{(n - k) \times n}$.
	Remark that $\G \H^{\top} = \0$ applies to any pair consisting of a generator and a parity-check matrix, and that both representations are not unique.
	The \emph{dual code} of $\mycode{C}$ is the code generated by $\H$.
	It has length $n$ and dimension $n-k$, and is denoted by $\mycode{C}^{\perp}$.
	
	\paragraph{Puncturing and shortening.}
	We will need two ways of deriving a shorter code from a given one (see e.g.\ in \cite[Sec.~1.5]{HuffmanPless2003FundamentalsErrorCorrecting}).
	First, we \emph{puncture} a linear code $\mycode{C} \subseteq \Fqm^{n}$ at position $i = 1, \dots, n$ by removing the $i$-th coordinate from each codeword.
	In other words,
	\begin{equation}
	    \puncCIdx{i} := \{(c_1, \dots, c_{i-1}, c_{i+1}, \dots, c_n) \mid (c_1, \dots, c_n) \in \mycode{C} \} \subseteq \Fqm^{n-1}.
	\end{equation}
	We obtain a generating matrix of $\puncCIdx{i}$ from a generator matrix of $\mycode{C}$ by removing the $i$-th column.
	Second, we \emph{shorten} $\mycode{C}$ at the $i$-th coordinate when we consider the subcode that vanishes at position $i$ and puncture it at position $i$.
	More precisely,
	\begin{equation}
	    \shortCIdx{i} := \{(c_1, \dots, c_{i-1}, c_{i+1}, \dots, c_n) \mid (c_1, \dots, c_{i-1}, 0, c_{i+1}, \dots, c_n) \in \mycode{C} \} \subseteq \Fqm^{n-1}.
	\end{equation}
	Both notions can be extended from one index to a set of coordinates in a straightforward manner.
	Further, it is worth noting that puncturing and shortening are dual operations in the sense that $(\puncCIdx{i})^{\perp} = \shortVarIdx{(\mycode{C}^{\perp})}{i}$ and $(\shortCIdx{i})^{\perp} = \puncVarIdx{(\mycode{C}^{\perp})}{i}$ apply according to \cite[Thm.~1.5.7]{HuffmanPless2003FundamentalsErrorCorrecting}.
	
	\paragraph{Different decoding metrics.}
	Codes can be endowed with different decoding metrics which measure e.g.\ the distance of an erroneous received vector from the closest codeword or the error weight.
	Classically, the \emph{Hamming weight} $\wtH(\v)$ of a vector $\v \in \Fqm^{n}$ is the number of its nonzero entries.
	For the \emph{rank weight}, we use the fact that $\Fqm$ is an $m$-dimensional vector space over its subfield $\Fq$.
	We define $\wtRk(\v)$ as the maximum number of $\Fq$-linearly independent entries of $\v$ and note that this is precisely the rank of the matrix representation of $\v$ over $\Fq$.
	The \emph{sum-rank weight} is a mixture between the Hamming weight and the rank weight.
	It is defined for vectors $\v \in \Fqm^{\n}$ which are divided into $\shots$ blocks of lengths $n_1, \dots, n_{\shots}$ and is given as
	\begin{equation}
		\wtSrk(\v) = \sum_{i=1}^{\shots} \wtRk(\shot{\v}{i}).
	\end{equation}
	Note that the sum-rank weight depends on the chosen block partition but we omit that fact from the notation for simplicity.
	The case $n_1 = n$ results in one block and recovers the rank metric, whereas the case of $n$ length-one blocks, i.e., $n_1 = \dots = n_n = 1$, corresponds to the Hamming metric \cite{Martinez-Penas2018SkewLinearizedReed}.
	The \emph{skew weight} depends on $\aut$ and on a $\PAut$-independent vector $\vecalpha \in \Fqm^n$.
	Its definition can be found in e.g.\ \cite{BoucherNouetowa2025DecodingAlgorithmSkew} and reads
	\begin{equation}
		\wtSkewIdx{\aut, \vecalpha}(\v) = \deg\bigl( \lclm_{\mystack{i = 1, \dots, n}{v_i \neq 0}} (x - \ConjAut{\alpha_i}{v_i}) \bigr).
	\end{equation}
	
	Each of the above discussed weights allows to obtain a metric on $\Fqm^n$ by letting the distance between two vectors $\v, \w \in \Fqm^n$ equal the weight of their difference $\v - \w$.
	The \emph{minimum distance} $d$ of a linear code $\mycode{C}$ with respect to a certain metric is the minimum weight of a nonzero codeword in the corresponding metric.
	It determines the error-correction capability $\left\lfloor \tfrac{d-1}{2} \right\rfloor$ for unique decoding with respect to errors in the respective metric.
	The Singleton bound relates the length $n$, the dimension $k$, and the minimum Hamming distance $\dH$ of a linear code via $d \leq n - k + 1$.
	Every code with minimum distance precisely $n - k + 1$ is called \emph{\acf{MDS}}.
	There are analog bounds for the minimum distances $\dRk$, $\dSrk$, and $\dSkewIdx{\;\aut, \vecalpha}$ with respect to rank, sum-rank, and skew metric.
	For $\Fqm$-linear codes, the upper bound is $n - k + 1$ in each case, and linear codes achieving these bounds have \emph{\ac{MRD}}, \emph{\ac{MSRD}}, and \emph{\ac{MSD}} with respect to $\aut$ and $\vecalpha$, respectively \cite{Martinez-Penas2018SkewLinearizedReed}.
	
	Later, we will consider $\Fqm$-linear \emph{isometries} for the Hamming metric, i.e., bijective linear maps that preserve the Hamming weight of vectors in $\Fqm^n$.
	The application of any such mapping to a code $\mycode{C} \subseteq \Fqm^n$ yields a \emph{monomially equivalent} code $\mycode{C}'$.
	The transition from $\mycode{C}$ to $\mycode{C}'$ corresponds to multiplying a generator matrix of $\mycode{C}$ by a monomial matrix from the right.
	Recall that a \emph{monomial matrix} $\M \in \Fqm^{n\times n}$ can be expressed as the product of a diagonal matrix $\diag(\d) \in \Fqm^{n \times n}$ with nonzero diagonal entries $\d \in \Fqm^{n}$ and a permutation matrix $\P \in \Fqm^{n \times n}$ containing precisely one one in each column and in each row and zeros elsewhere \cite[Sec.~1.7]{HuffmanPless2003FundamentalsErrorCorrecting}.

	\section{Skew evaluation codes}
	\label{sec:skew_eval_codes}
	
	This section contains definitions and basic properties of \ac{GSRS} and \ac{GLRS} codes and their special cases.
	Both code families are obtained as evaluation codes from skew polynomial rings with respect to remainder and generalized operator evaluation, respectively.
	That is why we use the term skew evaluation codes to represent all these code classes.
	Note however that the same or a similar term was used in the literature for slightly different codes \cite{BoucherUlmer2014LinearCodesUsing,LiuManganielloEtAl2015ConstructionDecodingGeneralized,Liu2016GeneralizedSkewReed}.
	We will discuss the differences after we have given the necessary background to understand them.
	
	\subsection{Generalized skew \ReSo codes and their duals}
	
	We start with the definition of \acf{GSRS} codes, which were e.g.\ studied in \cite{BoucherUlmer2014LinearCodesUsing,LiuManganielloEtAl2015ConstructionDecodingGeneralized,Liu2016GeneralizedSkewReed}.
	
	\begin{definition}[{GSRS codes}]\label{defi:GSRS}
		Consider a $\PAut$-independent vector $\vecalpha \in \Fqm^n$ and column multipliers $\veclambda \in \Fqm^n$ with only nonzero entries.
		Then, the \acf{GSRS} code of length $n$ and dimension $k$ is defined as
		\begin{equation}
			\genSkewRS{\vecalpha, \veclambda; n, k}_\aut := \bigl\{ \remev{f}{\vecalpha} \star \veclambda : f \in \SkewPolyringZeroDer \text{ with } \deg(f) < k \bigr\} \subseteq \Fqm^n.
		\end{equation}
		We call $\vecalpha$ its code locators and $\veclambda$ its column multipliers.
		Further, trivial column multipliers $\veclambda = \1$ yield a \ac{SRS} code which we denote by $\skewRS{\vecalpha; n, k}_\aut$.
	\end{definition}
	
	\ac{GSRS} codes have a generator matrix of a particularly nice form.
	Namely,
    \begin{equation}
        \G = \skewVandermonde{k}{\vecalpha} \cdot \diag(\veclambda):=\begin{pmatrix}
			\alpha_{1}^{\doubleBrackAut{0}} & \cdots & \alpha_{n}^{\doubleBrackAut{0}} \\
			\alpha_{1}^{\doubleBrackAut{1}} & \cdots & \alpha_{n}^{\doubleBrackAut{1}} \\
			\vdots & \vdots & \vdots \\
			\alpha_{1}^{\doubleBrackAut{k - 1}} & \cdots & \alpha_{n}^{\doubleBrackAut{k - 1}}
		\end{pmatrix}\cdot \diag(\veclambda)
    \end{equation}
    generates $\genSkewRS{\vecalpha, \veclambda; n, k}_\aut$, where $\diag(\veclambda)$ denotes the diagonal matrix having $\lambda_1,\dots,\lambda_n$ on the diagonal.
    In general, we call $\skewVandermonde{k}{\vecalpha}$ a \emph{skew Vandermonde matrix} and it has full rank $k \leq n$ if the entries of $\vecalpha$ are $\PAut$-independent (see \cite[Thm.~4.5]{LamLeroy1988VandermondeWronskianMatrices} and \cref{lem:p-indep_conditions}).

	It is known that \ac{GSRS} codes have minimum Hamming distance $d = n - k + 1$ and are thus \ac{MDS} codes~\cite[Thm.~4.1.4]{Liu2016GeneralizedSkewReed}.
	Particular parameter choices result in properties with respect to the skew metric, as was shown in \cite[Thm.~2]{BoucherNouetowa2025DecodingAlgorithmSkew}:
	
	\begin{lemma}[Skew-metric \ac{GSRS} codes]\label{lem:skew-metric-gsrs}
	    Let $\vecalpha, \veclambda \in \Fqm^n$ be \ac{GSRS} parameters according to \cref{defi:GSRS}.
	    Assume further that $\ConjAut{\vecalpha}{\veclambda} \in \Fqm^{n}$ is $\PAut$-independent.
		Then, the \ac{GSRS} code $\genSkewRS{\ConjAut{\vecalpha}{\veclambda}, \veclambda; n, k}_\aut$ is \ac{MSD} with respect to $\aut$ and $\vecalpha$.
	\end{lemma}
	
	Recently, a formula for the dual of an \ac{SRS} code was given in \cite[Thm.~3]{BoucherNouetowa2025DecodingAlgorithmSkew}.
	We can extend this result to generalized codes with column multipliers as follows:
	
	\begin{lemma}[{GSRS duals}]
		\label{lem:gsrs_duals}
		Consider \ac{GSRS} parameters $\vecalpha, \veclambda \in \Fqm^n$ as in \cref{defi:GSRS}.
		Moreover, let $\v \in \Fqm^n$ be a nonzero vector in $\genSkewRS{\vecalpha, \veclambda; n, n - 1}_{\aut}^{\perp}$ and assume that $\ConjAutinv{\autinv(\vecalpha)}{\v \star \veclambda}$ is $\PAutinv$-independent.
		Then,
		\begin{equation}
			\genSkewRS{\vecalpha, \veclambda; n, k}_{\aut}^{\perp} = \genSkewRS{\ConjAutinv{\autinv(\vecalpha)}{\v \star \veclambda}, \v; n, n - k}_{\autinv}.
		\end{equation}
	\end{lemma}
	
	\begin{proof}
		The codes $\genSkewRS{\vecalpha, \veclambda; n, k}_{\aut}$ and $\genSkewRS{\ConjAutinv{\autinv(\vecalpha)}{\v \star \veclambda}, \v; n, n - k}_{\autinv}$ have generator matrices of the form $\G = \skewVandermonde{k}{\vecalpha} \cdot \diag(\veclambda) \in \Fqm^{k \times n}$ and $\H = \skewVandermondeInv{n-k}{\ConjAutinv{\autinv(\vecalpha)}{\v \star \veclambda}} \cdot \diag(\v) \in \Fqm^{(n - k) \times n}$, respectively.
		We first show that $\G \H^{\top} = \0$ holds, i.e., that the \ac{GSRS} code on the right-hand side is a subcode of the wanted dual.
		This is the case if and only if for all $g = 1, \dots, k$ and all $h = 1, \dots, n - k$
		\begin{equation}
			\bigl( \vecalpha^{\doubleBrackAut{g - 1}} \star \veclambda \bigr) \cdot \bigl( \ConjAutinv{\autinv(\vecalpha)}{\v \star \veclambda}^{\doubleBrackAutinv{h - 1}} \star \v \bigr)^{\top} = 0
		\end{equation}
		applies.
		This equality can be rephrased as
		\begin{align}
		    0 &= \sum_{j = 1}^{n} \alpha_{j}^{\doubleBrackAut{g - 1}} \autinv(\alpha_j)^{\doubleBrackAutinv{h - 1}} \autinv(v_j \lambda_j)^{\doubleBrackAutinv{h - 1}} (v_{j}^{-1} \lambda_{j}^{-1})^{\doubleBrackAutinv{h - 1}} v_j \lambda_j
			\\
		    &= \sum_{j = 1}^{n} \aut^{- (h - 1)}\bigl( \alpha_{j}^{\doubleBrackAut{g + h - 2}} \bigr) \aut^{- (h - 1)}(v_j \lambda_j)
			= \aut^{- (h - 1)}\Bigl( \sum_{j = 1}^{n} \alpha_{j}^{\doubleBrackAut{g + h - 2}} v_j \lambda_j \Bigr)
		\end{align}
		and is satisfied if and only if $\sum_{j = 1}^{n} \alpha_{j}^{\doubleBrackAut{e - 1}} v_j \lambda_j = 0$ holds for all $e = 1, \dots, n - 1$.
		In other words, it holds for any $\v \in \genSkewRS{\vecalpha, \veclambda; n, n - 1}_{\aut}^{\perp}$ as assumed in the statement.
		
		The equality of the two codes follows from the assumption that the entries of $\ConjAutinv{\autinv(\vecalpha)}{\v \star \veclambda}$ are $\PAutinv$-independent which is sufficient for the skew Vandermonde matrix $\skewVandermondeInv{n - k}{\ConjAutinv{\autinv(\vecalpha)}{\v \star \veclambda}}$ to have rank $n-k$.
		Further, it can be shown with a similar argument as in \cite[Thm.~4]{Martinez-PenasKschischang2019ReliableSecureMultishot} and \cite[Lem.~2]{BoucherNouetowa2025DecodingAlgorithmSkew} that every $\v \in \genSkewRS{\vecalpha, \veclambda; n, n - 1}_{\aut}^{\perp}$ contains only nonzero elements.
		This ensures that any generator matrix of $\genSkewRS{\ConjAutinv{\autinv(\vecalpha)}{\v \star \veclambda}, \v; n, n - k}_{\autinv}$ has rank $n-k$, i.e., the code on the right-hand side has dimension $n-k$.
		\qed
	\end{proof}

	\subsection{Generalized linearized \ReSo codes and their duals}
	
	\Ac{LRS} codes were introduced in \cite[Def.~31]{Martinez-Penas2018SkewLinearizedReed} and generalized to block multipliers in \cite{HoermannBartzEtAl2023DistinguishingRecoveringGeneralized}.
	We extend the definition to arbitrary nonzero column multipliers:
	
	\begin{definition}[GLRS codes]
		\label{def:glrs_code}
		Consider a vector $\b \in \Fqm^{\n}$ with $\wtSrk(\b) = n$ and let $\veclambda \in \Fqm^{\n}$ contain only nonzero elements.
		Further, let $a_1, \dots, a_{\shots} \in \Fqm$ belong to distinct nontrivial
		$\aut$-conjugacy classes of $\Fqm$ and let $\a \in \Fqm^{\n}$ be the corresponding $\n$-vector as defined above \cref{lem:p-indep_conditions}. %
		Then, the {\acf{GLRS} code} of length $n$ and dimension $k$ is
		\begin{equation}
			\GLRS{\b, \a, \veclambda; n, k}_\aut \defeq
			\bigl\{ \opev{f}{\b}{\a} \star \veclambda :
			f \in \SkewPolyringZeroDer \text{ with } \deg(f) < k \bigr\}
			\subseteq \Fqm^{\n}.
		\end{equation}
		We call $\b$ its code locators, $\a$ its evaluation parameters, and $\veclambda$ its column multipliers.
		If $\veclambda = \1$, we obtain an \ac{LRS} code and denote it by $\linRS{\b, \a; n, k}_\aut$.
	\end{definition}

	The code $\GLRS{\b, \a, \veclambda; n, k}_\aut$ has a generator matrix of the form
	\begin{equation}\label{eq:GLRS_gen_mat}
		\G = \opVandermonde{k}{\b}{\a} \cdot \diag(\veclambda)
	\end{equation}
	where the blocks of the \emph{generalized Moore matrix} $\opVandermonde{k}{\b}{\a} \in \Fqm^{k \times \n}$ are defined as
	\begin{equation}
		\opVandermonde{k}{\shot{\b}{i}}{a_i} =
	    \begin{pmatrix}
	        \subShot{b}{1}{i} & \dots & \subShot{b}{\lenShot{i}}{i} \\
	        \opfull{a_i}{\subShot{b}{1}{i}} & \dots & \opfull{a_i}{\subShot{b}{\lenShot{i}}{i}} \\
	        \vdots & \dots & \vdots \\
	        \opfullexp{a_i}{\subShot{b}{1}{i}}{k-1} & \dots & \opfullexp{a_i}{\subShot{b}{\lenShot{i}}{i}}{k-1}
	    \end{pmatrix}
		\quad \text{for all } i = 1, \dots, \shots.
	\end{equation}
	
	Since the multiplication of each block of a code with a nonzero $\Fqm$-element or with a full-rank $\Fq$-matrix is a sum-rank isometry \cite{Martinez-Penas2020HammingSimplexCodes,AlfaranoLobilloEtAl2022SumRankProduct}, \ac{GLRS} codes with the following parameter restrictions are \ac{MSRD}:
	
	\begin{lemma}[Sum-rank-metric \ac{GLRS} codes]
		\label{lem:sr-glrs-codes}
		Consider \ac{GLRS} parameters $\a, \b, \veclambda \in \Fqm^{\n}$ as in \cref{def:glrs_code}.
	    If $\subShot{\lambda}{1}{i} = \dots = \subShot{\lambda}{\lenShot{i}}{i}$ applies for all $i = 1, \dots, \shots$ or if $\veclambda \in \Fq^n$, the code $\GLRS{\b, \a, \veclambda; n, k}_\aut$ is \ac{MSRD}.
	\end{lemma}
	
	We can now derive a statement about \ac{GLRS} duals similar to \cref{lem:gsrs_duals} in the \ac{GSRS} setting.
	We obtain the following:
	
	\begin{lemma}[GLRS duals]
		\label{lem:glrs_duals}
		Let $\a, \b, \veclambda \in \Fqm^{\n}$ be \ac{GLRS} parameters according to \cref{def:glrs_code}.
		Further, choose any nonzero vector $\v \in \GLRS{\b, \a, \veclambda; n, n - 1}_{\aut}^{\perp}$ and assume that $\veclambda \star \v$ has sum-rank weight $n$.
		Then, it holds
	    \begin{equation}
	        \GLRS{\b, \a, \veclambda; n, k}_{\aut}^{\perp} = \GLRS{\v \star \veclambda, \autinv(\a), \veclambda^{-1}; n, n - k}_{\autinv}.
	    \end{equation}
	\end{lemma}
	
	\begin{proof}
		This proof is an adaptation of the one of \cite[Thm.~4]{Martinez-PenasKschischang2019ReliableSecureMultishot} and follows the same strategy as the proof of \cref{lem:gsrs_duals}.
		For the orthogonality of the codes, we show that $\G = \opVandermonde{k}{\b}{\a} \cdot \diag(\veclambda)$ and $\H = \opVandermondeInv{n-k}{\v \star \veclambda}{\autinv(\a)} \cdot \diag(\veclambda^{-1})$ satisfy $\G \H^{\top} = \0$.
		In other words, we require that for all $g = 1, \dots, k$ and all $h = 1, \dots, n - k$
		\begin{equation}
			\bigl( \opfullexp{\a}{\b}{g-1} \star \veclambda \bigr) \cdot \bigl( \opfullexpinv{\autinv(\a)}{\v \star \veclambda}{h-1} \star \veclambda^{-1} \bigr)^{\top} = 0
		\end{equation}
		holds.
		With similar steps as in \cref{lem:gsrs_duals}, we rephrase this equivalently as
		\begin{equation}
			\aut^{h-1} \Bigl( \sum_{j = 1}^{n} \opfullexp{a_j}{b_j}{g+h-2} v_j \lambda_j \Bigr) = 0
		\end{equation}
		and finally as $\sum_{j = 1}^{n} \opfullexp{a_j}{b_j}{e-1} v_j \lambda_j = 0$ for all $e = 1, \dots, n-1$.
		By definition, this is true for any $\v \in \GLRS{\b, \a, \veclambda; n, n - 1}_{\aut}^{\perp}$ and we hence proved that $\GLRS{\v \star \veclambda, \autinv(\a), \veclambda^{-1}; n, n - k}_{\autinv} \subseteq \GLRS{\b, \a, \veclambda; n, k}_{\aut}^{\perp}$ applies.
		
		Since $\autinv(\a)$ contains representatives of distinct nontrivial $\autinv$-conjugacy classes and we assumed that $\wtSrk(\v \star \veclambda) = n$ applies, $\opVandermondeInv{n-k}{\v \star \veclambda}{\autinv(\a)}$ has full $\Fqm$-rank $n - k$ according to \cite[Thm. 2]{Martinez-Penas2018SkewLinearizedReed} and \cite[Thm. 4.5]{LamLeroy1988VandermondeWronskianMatrices}.
		Further, $\veclambda$ and thus also $\veclambda^{-1}$ contains only nonzero entries and consequently $\H = \opVandermondeInv{n-k}{\v \star \veclambda}{\autinv(\a)}$ has full rank.
		This shows that the \ac{GLRS} code on the right-hand side of the statement has dimension $n-k$, which finishes the proof.
		\qed
	\end{proof}
	
	The nonzero vector $\v$ that is arbitrarily chosen from $\GLRS{\b, \a, \veclambda; n, n - 1}_{\aut}^{\perp}$ in the above theorem actually satisfies $\wtSrk(\v) = n$, which follows from an adaptation of the proof of \cite[Thm.~4]{Martinez-PenasKschischang2019ReliableSecureMultishot}.
	Thus, the condition $\wtSrk(\v \star \veclambda) = n$ from the theorem statement is automatically satisfied when $\veclambda \in \Fqm^{n}$ satisfies any of the two conditions given in \cref{lem:sr-glrs-codes}.
	This means that the sum-rank-metric \ac{GLRS} codes defined in \cref{lem:sr-glrs-codes}, i.e., the ones with $\Fqm$-block multipliers or $\Fq$-column multipliers, are closed under duality.
	In particular, this holds for \ac{LRS} codes and generalizes the known results on \ac{LRS} duals from \cite[Thm.~4]{Martinez-PenasKschischang2019ReliableSecureMultishot}.

	\subsection{Connections between GSRS, GLRS, GRS, and Gabidulin codes}
	\label{subsec:code_connections}
	
	It is well-known that there is an isometry which maps \ac{GSRS} codes in the realm of the skew metric to \ac{GLRS} codes in the sum-rank world \cite{Martinez-Penas2018SkewLinearizedReed}.
	We now give equalities between codes of the two families which are independent of the considered metrics.
	We believe that the resulting explicit transformations of the code parameters make the results more accessible.
	
	\begin{lemma}[From GLRS to GSRS codes]
		\label{lem:glrs-to-gsrs}
		For \ac{GLRS} parameters $\a, \b, \veclambda \in \Fqm^{\n}$ as in \cref{def:glrs_code}, it holds
		\begin{equation}
		    \GLRS{\b, \a, \veclambda; n, k}_{\aut} = \genSkewRS{\ConjAut{\a}{\b}, \veclambda \star \b; n, k}_\aut.
		\end{equation}
	\end{lemma}
	
	\begin{proof}
	    Recall that each codeword $\c \in \GLRS{\b, \a, \veclambda; n, k}_\aut$ corresponds to a skew polynomial $f \in \SkewPolyringZeroDer$ of degree less than $k$ via $\c = \opev{f}{\b}{\a} \star \veclambda$.
		Since equation \eqref{eq:connection_remainder_gen-op} implies $\lambda \opev{f}{b}{a} = \lambda \remev{f}{\ConjAut{a}{b}} b$ for any $a, b, \lambda \in \Fqm^{\ast}$, the statement follows.
		\qed
	\end{proof}
	
	\begin{lemma}[From GSRS to GLRS codes]
		\label{lem:gsrs-to-glrs}
		Consider \ac{GSRS} parameters $\vecalpha, \veclambda \in \Fqm^{n}$ satisfying the conditions of \cref{defi:GSRS}.
		Let further $\a, \b \in \Fqm^{\n}$ be the outcome of \cref{lem:p-indep_conditions} such that $\pi(\vecalpha) = \ConjAut{\a}{\b}$ applies.
		Then,
		\begin{equation}
		    \genSkewRS{\pi(\vecalpha), \veclambda; n, k}_\aut = \GLRS{\b, \a, \veclambda \star \b^{-1}; n, k}_\aut.
		\end{equation}
	\end{lemma}
	
	\begin{proof}
	    This follows from the same argument as \cref{lem:glrs-to-gsrs}.
		\qed
	\end{proof}
	
	The permutation $\pi$ in the above \cref{lem:gsrs-to-glrs} is no restriction for the choice of code locators $\vecalpha$ of the \ac{GSRS} code that is transformed.
	In fact, recall from \cref{rem:sum-rank_blocks} that $\pi$ only groups the entries of $\vecalpha$ by conjugacy class which is convenient for the \ac{GLRS} representation and the related sum-rank metric.
	It is however not necessary: the permuted \ac{GLRS} code with code locators $\pi^{-1}(\b)$ and evaluation parameters $\pi^{-1}(\a)$ has the same Hamming-metric properties.
	Further, its sum-rank properties carry over if the sum-rank weight is defined in terms of suitable index sets instead of a block partition $\n$.
	
	Next, we deal with prominent special cases of \ac{GSRS} and \ac{GLRS} codes: \ac{GRS} and Gabidulin codes.
	We state comprehensively how they can be obtained from the \ac{GSRS} and the \ac{GLRS} framework and thus give a good starting point to compare our distinguishing results to known distinguishers for these code families later.
	
    \paragraph{GRS codes.}
	Let us first recall the definition of a \acf{GRS} code of length $n$ and dimension $k$, which reads
	\begin{align}
	    \GRS{\b, \veclambda; n, k} := &\left\{ f(\b) \star \veclambda : f \in \Polyring \text{ with } \deg(f) < k \right\} \subseteq \Fqm^n
	\end{align}
	for distinct code locators $\b \in \Fqm^n$ and nonzero column multipliers $\veclambda \in \Fqm^n$.
	Observe that this definition employs conventional polynomials from $\Polyring$ and their standard evaluation.
	We can express $\GRS{\b, \veclambda; n, k}$ as a \ac{GSRS} and as a \ac{GLRS} code with block partition $\n = \1 \in \Fqm^{n}$ as follows:
	\begin{equation}\label{eq:grs_as_gsrs_glrs}
	    \GRS{\b, \veclambda; n, k}
	    = \genSkewRS{\b, \veclambda; n, k}_{\id}
	    = \GLRS{\1, \b, \veclambda; n, k}_{\id}.
	\end{equation}
	The chosen automorphism is the identity and the equalities follow from expressing the remainder evaluation and the generalized operator evaluation of any $f(x) = \sum_i f_i x^i \in \Fqm[x;\id]$ as
	\begin{gather}
		\remev{f}{b} = \sum\nolimits_i f_i b^{\doubleBrackWIndex{i}{\id}} = \sum\nolimits_i f_i b^i
		\\
		\text{and} \quad \opev{f}{1}{b} = \sum\nolimits_i f_i \opexp{1}{b}{i} = \sum\nolimits_i f_i \id^{i}(1) b^{\doubleBrackWIndex{i}{\id}} = \sum\nolimits_i f_i b^i
	\end{gather}
	for any $b \in \Fqm$, respectively.
	The code locators $\b$ of the \ac{GRS} code and its \ac{GSRS} representation actually become the evaluation parameters in its \ac{GLRS} representation.
	In contrast, the code locators of Gabidulin codes stay code locators also in the \ac{GLRS} representation, as we will see shortly.
	
	Observe also that the distinctness of $b_1, \dots, b_n$ translates to the required conditions in the \ac{GSRS} and the \ac{GLRS} case:
	Since the identity creates $q^m - 1$ nontrivial conjugacy classes of cardinality one, every vector with distinct nonzero entries such as $\b$ is $\PIdx{\id}$-independent.
	Further, $b_1, \dots, b_n$ are used as evaluation parameters in the \ac{GLRS} setting where they represent different nontrivial conjugacy classes due to the reasoning above.
	In addition, the all-one vector has sum-rank weight $n$ as required because each block has length precisely one.
	
    \paragraph{Gabidulin codes.}
	In the classical definition of Gabidulin codes, we pick $\Fq$-linearly independent code locators $\b \in \Fqm^n$ and define the respective code of length $n$ and dimension $k$ as
	\begin{align}\label{eq:gab-def}
	    \Gab{\b; n, k} :=& \bigl\{ f(\b) : f \in \Linpolyring \text{ with } \qdeg(f) < k \bigr\} \subseteq \Fqm^n,
	\end{align}
	where $\Linpolyring$ is the ring of linearized polynomials.
	A linearized polynomial has the form $f(x) = \sum_i f_i x^{q^i} = \sum_i f_i x^{\brackWIndex{i}{\frob}}$ with finitely many $\Fqm$-coefficients, and its evaluation at any $b \in \Fqm$ is given as $f(b) = \sum_i f_i b^{\brackWIndex{i}{\frob}}$.
	The $q$-degree $\qdeg(f)$ equals the maximum $i$ for which $f_i \neq 0$ applies and $-\infty$ if $f = 0$.

	We use the Frobenius automorphism $\sigma$ to express $\Gab{\b; n, k}$ as \ac{GSRS} and as \ac{GLRS} code with block partition $\n = (n) \in \Fqm$.
	Namely,
	\begin{align}\label{eq:gab_as_gsrs_glrs}
	     \Gab{\b; n, k}
	     = \genSkewRS{\b^{q-1}, \b; n,k}_{\sigma}
	     = \GLRS{\b, (1), \1; n, k}_{\sigma}.
	\end{align}
	Similar to the case of \ac{GRS} codes above, the equalities follow from the following facts about remainder and generalized operator evaluation of a skew polynomial $f(x) = \sum_i f_i x^i \in \Fqm[x;\frob]$ at any $b \in \Fqm$:
	\begin{gather}
	    b \remev{f}{b^{q-1}} = b \sum\nolimits_i f_i b^{(q-1) \cdot \doubleBrackWIndex{i}{\frob}}
	    = \sum\nolimits_i f_i b b^{q^i -1} = \sum\nolimits_i f_i b^{\brackWIndex{i}{\frob}}
		\\
		\text{and} \quad \opev{f}{b}{1} = \sum\nolimits_i f_i \opexp{b}{1}{i} = \sum\nolimits_i f_i \frob^i(b) 1^{\doubleBrackWIndex{i}{\frob}} = \sum\nolimits_i f_i b^{\brackWIndex{i}{\frob}}.
	\end{gather}
	Observe that the column multipliers in the \ac{GSRS} representation cannot be chosen freely in the Gabidulin scenario but they depend heavily on the code locators.

	The condition $\wtRk(\b) = n$ implies that the column multipliers of the \ac{GSRS} representation are nonzero.
	Further, $\b^{q-1} = \ConjFrob{\1}{\b}$ applies and thus $\b^{q-1}$ is $\PIdx{\frob}$-independent according to \cref{lem:p-indep_conditions} as demanded.
	The requirements for the respective \ac{GLRS} parameters are clearly satisfied.
	
	\begin{remark}
		Note that other works in the literature use the term skew evaluation codes slightly differently.
		In \cite{LiuManganielloEtAl2015ConstructionDecodingGeneralized,Liu2016GeneralizedSkewReed}, \ac{GSE} codes are a superclass of \ac{GSRS} codes.
		The relaxation lies in the fact that the code locators $\vecalpha \in \Fqm^{n}$ do not need to be $\PAut$-independent but only have to ensure that $\skewVandermonde{k}{\vecalpha}$ has full rank.
		In \cite{BoucherUlmer2014LinearCodesUsing}, remainder evaluation skew codes are the same as the \ac{GSE} codes from \cite{LiuManganielloEtAl2015ConstructionDecodingGeneralized,Liu2016GeneralizedSkewReed} but without column multipliers.
		There, also operator evaluation skew codes are studied.
		They are related to generalized operator evaluation but are not exactly the same as \ac{LRS} codes.
		In particular, their generator matrix is the Wronskian matrix and not the generalized Moore matrix.
	\end{remark}

	\subsection{Code distinguishability and McEliece-like cryptosystems}
	\label{sec:distinguishers_and_crypto}
	
	In this work, we study the algebraic structure of skew evaluation codes with a particular focus on understanding if these code families are efficiently distinguishable from random codes.
	A distinguisher for \ac{GSRS}/\ac{GLRS} codes is an algorithm that is capable of correctly answering the following question with a higher success probability than random guessing:
	Was a given matrix $\G \in \Fqm^{k \times n}$ drawn randomly from the set of all full-rank matrices in $\Fqm^{k \times n}$ or from the set of all generator matrices of \ac{GSRS}/\ac{GLRS} codes of length $n$ and dimension $k$?
	Remark that the notion of a distinguisher can be formalized in the form of a game, and put into relation with attackers of the underlying McEliece-like problem.
	But since the formalism seems not to be beneficial for the presentation of the present work, we point the interested reader to \cite[Sec.~II]{FaugereGauthier-UmanaEtAl2013DistinguisherHighRate}, for example.
	
	Our motivation clearly comes from code-based cryptography where the indistinguishability assumption for the code class used in the McEliece-like setting is commonly seen as a requirement for security.
	To be more precise, system designers usually aim for \ac{IND-CCA2} in accordance with the standard requirements from \ac{NIST}'s \ac{PQC} standardization for general-use \acp{KEM}.
	This can be achieved, as in the case of Classic McEliece, by applying an \ac{IND-CCA2} conversion to a \ac{PKE} which is \ac{OW-CPA}.
	For example, Kobara and Imai investigate the applicability of several known generic transformations and derive alternatives tailored to the McEliece setting in the \ac{ROM} in~\cite{KobaraImai2001SemanticallySecureMceliece}.
	
	Up to our knowledge, there are two ways in which the \acs{OW-CPA} security of McEliece-like cryptosystems can be attacked: by means of key-recovery attacks, in which the adversary retrieves an equivalent secret key from the sheer knowledge of the public key, or by means of generic decoding, in which the adversary ignores the secret code structure and tries to decode an obtained ciphertext as if it originated from a random code.
	The latter can be formalized as the \ac{SDP} whose decisional version was shown to be NP-complete in~\cite{BerlekampMcElieceEtAl1978InherentIntractabilityCertain} for the binary field and in~\cite{Barg1994SomeNewNp} for arbitrary finite fields.
	The former structural attacks correspond to the distinguishability problem for the used code family.
	The hardness of these two problems implies \ac{OW-CPA} security but is however not equivalent to it.
	It is nevertheless common to study them instead of the \ac{OW-CPA} security itself.
	
	Observe that a McEliece-like system can remain secure even if an efficient distinguisher breaks the indistinguishability assumption.
	This only turns into a security risk if the distinguisher can be extended to a key-recovery attack.
	In contrast, every successful break of the \ac{OW-CPA} security yields a distinguisher for the public keys or solves the \ac{SDP}.

	\section{GRS subcodes of skew evaluation codes}
	\label{sec:grs_subcodes}
	
	In addition to their inherent algebraic structure arising from being evaluation codes of skew polynomials, generalized skew and linearized Reed–Solomon codes also exhibit a rich compositional property: they can be expressed as a direct product of several classical generalized Reed--Solomon (GRS) codes.
	We establish this decomposition in the following and also explore its shape in the special case of Gabidulin and GRS codes.
	For the ease of notation, let us first define
	\begin{equation}
		\label{eq:def_ki}
        \kIdx{i} \defeq
	    \begin{cases}
	        \left\lceil \tfrac{k}{m} \right\rceil & \text{for } i = 1, \dots, (k \mod m) \\
	        \left\lfloor \tfrac{k}{m} \right\rfloor & \text{for } i = (k \mod m) + 1, \dots, m.
	    \end{cases}
	\end{equation}
	
    \begin{theorem}[GRS subcodes of GSRS codes]\label{thm:gsrs-decomp}
        Let $\vecalpha, \veclambda \in \Fqm^n$ be \ac{GSRS} parameters according to \cref{defi:GSRS}.
	    Then, it applies
	    \begin{equation}
			\genSkewRS{\vecalpha, \veclambda; n, k}_\aut = \bigoplus_{i=1}^{m}  \genRS{N_{\aut}(\vecalpha), \vecalpha^{\doubleBrackAut{i-1}} \star \veclambda;n, \kIdx{i}}
	    \end{equation}
        with $\kIdx{1}, \dots, \kIdx{m}$ from \eqref{eq:def_ki}.
		Note that the code locators $N_{\aut}(\vecalpha)$ of the GRS codes might not all be distinct.
    \end{theorem}
	
	\begin{proof}
	For simplicity we assume that $\theta=\sigma$ is the Frobenius automorphisms, and note that the same arguments in the proof can be used for general $\aut$. Moreover, we can take the column mutlipliers $\veclambda$ out, prove the statement for trivial column multipliers, and multiply the component codes by them again in the end.
	
	Consider the basis of $\genSkewRS{\vecalpha, \mathbf{1}; n, k}_\aut$ consisting of $\vecalpha^{\doubleBrackAut{x}}$ for $x=0,\dots,k-1$ and
	 note that for $0<x<k$
	\begin{align}
		\alpha_j^{\doubleBrackAut{x}}&= \alpha_j^{q^{x-1}+\dots+q+1}
		= \alpha_j^{(1+q+\dots+q^{m-1})\lfloor\frac{x-1}{m}\rfloor+1+q+\dots+q^{(x-1) \mod m}}\\
	&=  N_{\aut}(\alpha_j)^{\lfloor\frac{x-1}{m}\rfloor} \alpha_j^{\doubleBrackAut{x \mod m}}.
	\end{align}
	For $x=0$ we get that
	$$\alpha_j^{\doubleBrackAut{0}}=1=N_{\aut}(\alpha_j)^0 \alpha_j^{\doubleBrackAut{0}}.$$
	Hence, for $0\leq i < k-m\lfloor \frac{k}{m} \rfloor $, the basis vectors with exponents $\doubleBrackAut{i},\doubleBrackAut{i+m},\dots,$ $\doubleBrackAut{i+m(\lceil \frac{k}{m} \rceil-1)}$ generate the code
	$$ \genRS{N_{\aut}(\vecalpha), \vecalpha^{\doubleBrackAut{i}};n,\left\lceil \tfrac{k}{m} \right\rceil}.$$
	Similarly, the basis vectors with exponents $\doubleBrackAut{i},\doubleBrackAut{i+m},\dots,\doubleBrackAut{i+m(\lfloor \frac{k}{m} \rfloor-1)}$ for $k-m\lfloor \frac{k}{m} \rfloor \leq i <m$ generate
	$$ \genRS{N_{\aut}(\vecalpha), \vecalpha^{\doubleBrackAut{i}};n,\left\lfloor \tfrac{k}{m} \right\rfloor}.$$
	This proves that all claimed \ac{GRS} subcodes are contained in $\genSkewRS{\vecalpha, \veclambda; n, k}_\aut$.
	Finally, we show that the sum of the $m$ distinct \ac{GRS} codes is direct since their dimensions sum up precisely to $k$.
	This is easy to see in case $k$ divides $m$ since $\sum_{i=1}^{m} k_i = \sum_{i=1}^{m} \tfrac{k}{m} = k$.
	If $k$ does not divide $m$, we can split up the sum in the spirit of \eqref{eq:def_ki} and use the fact that $k \mod m = k - m \left\lfloor \tfrac{k}{m} \right\rfloor$ to obtain
	\begin{align}
		\sum_{i=1}^{m} k_i
		&= (k \mod m) \cdot \left\lceil \tfrac{k}{m} \right\rceil + (m - k \mod m) \cdot \left\lfloor \tfrac{k}{m} \right\rfloor
		\\
		&= \left(k - m \left\lfloor \tfrac{k}{m} \right\rfloor\right) \cdot \left(\left\lfloor \tfrac{k}{m} \right\rfloor + 1\right) + \left(m - k + m \left\lfloor \tfrac{k}{m} \right\rfloor\right) \cdot \left\lfloor \tfrac{k}{m} \right\rfloor
		= k.
	\end{align}
	This finishes the proof.
			\qed
	     \end{proof}
	
	As mentioned in \cref{thm:gsrs-decomp}, the resulting GRS subcodes are not always proper classical GRS codes as their code locators might not all be distinct.
	In fact, we can easily see how many distinct elements $N_{\aut}(\vecalpha)$ contains:
	Recall therefore that members of the same $\aut$-conjugacy class have the same norm according to Hilbert's Theorem 90.
	Thus, apply \cref{lem:p-indep_conditions} to obtain a representation $\pi(\vecalpha) = \ConjAut{\a}{\b}$ with $\a, \b \in \Fqm^{\n}$.
	Since $\vecalpha$ is $\PAut$-independent, $a_1, \dots, a_\shots$ belong to distinct nontrivial conjugacy classes and $N_{\aut}(a_1), \dots, N_{\aut}(a_\shots)$ are distinct.
	This means that $N_{\aut}(\vecalpha) = N_{\aut}(\a)$ contains precisely $n_i$  copies of $N_{\aut}(a_i)$ for each $i = 1, \dots, \shots$.
	As a consequence, this fact ensures that each of the GRS subcodes in \cref{thm:gsrs-decomp} attains the claimed dimension.
	The Vandermonde part of their generator matrices is $\skewVandermonde{\kIdx{i}}{N_{\aut}(\a)}$ for $\kIdx{i} \in \{\lceil \frac{k}{m} \rceil, \lfloor \frac{k}{m} \rfloor\}$ and has rank $\min(\shots, \kIdx{i})$.
	The minimum is ensured to always be $\kIdx{i}$ since $\tfrac{k}{m} \leq \tfrac{n}{m} \leq \tfrac{\shots m}{m} \leq \shots$.

    \begin{remark}\label{rem:k<m}
        In the case $k\leq m$, the decomposition presented above has exactly $k$ GRS components, all of which have dimension one.
    \end{remark}
	
	In the decomposition above, the component GRS codes share a common set of code locators, all contained in the base field $\Fq$, but differ in their respective sets of column multipliers. From a cryptographic perspective, the presence of these particular column multipliers inhibits straightforward recovery attacks of the Sidelnikov–Shestakov type.
	On the other hand, this structural regularity also has implications for the behavior of the code under the coordinate-wise (or star) product, particularly with respect to its square. Specifically, the structured nature of the component codes causes the dimension of the square code to deviate from that of a random linear code of the same length and dimension for many code parameters, as we will analyze in the next section.
	
	\begin{theorem}[GRS subcodes of GLRS codes]
		\label{thm:glrs-decomp}
		Let $\a, \b, \veclambda \in \Fqm^{\n}$ be \ac{GLRS} parameters according to \cref{def:glrs_code}.
		Then, it applies
		\begin{align}
		    \GLRS{\b, \a, \veclambda; n, k}_{\aut}
			&= \bigoplus_{i=1}^{m} \GRS{N_{\aut}(\a), \a^{\doubleBrackAut{i-1}} \star \veclambda \star \aut^{i-1}(\b);n, \kIdx{i}}
		\end{align}
		with $\kIdx{1}, \dots, \kIdx{m}$ being defined as in \eqref{eq:def_ki}.
		Note that the code locators $N_{\aut}(\a)$ of the GRS codes might not all be distinct.
	\end{theorem}
	
	\begin{proof}
		Recall that $\opfullexp{a}{b}{i} = \aut^i(b) a^{\doubleBrackAut{i}}$ applies for $i>0$ and for all $a, b \in \Fqm$, and that the code $\GLRS{\b, \a, \veclambda; n, k}_\aut$ has a generator matrix of the form $\G = \opVandermonde{k}{\b}{\a} \cdot \diag(\veclambda) \in \Fqm^{k \times \n}$ defined in \eqref{eq:GLRS_gen_mat}.
		Thus, we can express the $i$-th block of the Moore part $\opVandermonde{k}{\b}{\a}$ as
		\begin{align}
		    \opVandermonde{k}{\shot{\b}{i}}{\shot{\a}{i}} &=
			\begin{pmatrix}
			    \subShot{b}{1}{i} & \cdots & \subShot{b}{\lenShot{i}}{i} \\
				\opfull{a_i}{\subShot{b}{1}{i}} & \cdots & \opfull{a_i}{\subShot{b}{\lenShot{i}}{i}} \\
			    \vdots & \cdots & \vdots \\
				\opfullexp{a_i}{\subShot{b}{1}{i}}{k-1} & \cdots & \opfullexp{a_i}{\subShot{b}{\lenShot{i}}{i}}{k-1}
			\end{pmatrix}
		    \\
			&=
		    \begin{pmatrix}
			    \subShot{b}{1}{i} \cdot {a_i}^{\doubleBrackAut{0}} & \cdots & \subShot{b}{\lenShot{i}}{i} \cdot {a_i}^{\doubleBrackAut{0}} \\
				\aut(\subShot{b}{1}{i}) \cdot {a_i}^{\doubleBrackAut{1}} & \cdots & \aut(\subShot{b}{\lenShot{i}}{i}) \cdot {a_i}^{\doubleBrackAut{1}} \\
				\vdots & \cdots & \vdots \\
				\aut^{k-1}(\subShot{b}{1}{i}) \cdot {a_i}^{\doubleBrackAut{k - 1}} & \cdots & \aut^{k-1}(\subShot{b}{\lenShot{i}}{i}) \cdot {a_i}^{\doubleBrackAut{k - 1}}
		    \end{pmatrix}.
		\end{align}
		Ignoring the parts depending on $\shot{\b}{i}$, this matrix looks like the skew Vandermonde matrix $\skewVandermonde{k}{\shot{\a}{i}}$, i.e., similar to the \ac{GSRS} case.
		Recall that \cref{thm:gsrs-decomp} combines the rows whose indices coincide modulo $m$ into one \ac{GRS} subcode.
		Since $\aut^{m} = \id$ holds, the factor in the $j$-th row and $l$-th column is $\aut^{j \mod m}(\subShot{b}{l}{i})$.
		In other words, this factor is the same for a fixed column and for rows with indices coinciding modulo $m$.
		It can thus be moved into the column multipliers of the respective \ac{GRS} subcode and we obtain the theorem's statement with the same strategy that we used in the \ac{GSRS} setting in \cref{thm:gsrs-decomp}.
		\qed
	\end{proof}
	
	\paragraph{Gabidulin codes.}
	Let us next look at how the decompositions behave for the special cases of Gabidulin codes.
	Note therefore that the required $\Fq$-linear independence of the $n$ code locators from $\Fqm$ implies that $k\leq n \leq m$ holds.
	
	\begin{corollary}[GRS subcodes of Gabidulin codes]
		Let all entries of the vector $\b \in \Fqm^n$ be $\Fq$-linearly independent.
		Then,
	    \begin{equation}
	        \Gab{\b; n, k}
			= \bigoplus_{i=1}^{k} \GRS{\1, \frob^{i-1}(\b); n, 1}.
	    \end{equation}
	\end{corollary}
	
	\begin{proof}
		If we express $\Gab{\b; n, k}$ as \ac{GSRS} code as in \eqref{eq:gab_as_gsrs_glrs} and apply \cref{thm:gsrs-decomp}, we obtain
		\begin{align*}
		    		    \Gab{\b; n, k}
			&= \bigoplus_{i=1}^{m} \GRS{N_{\aut}(\b^{q-1}), (\b^{q-1})^{\doubleBrackWIndex{i-1}{\frob}} \star \b; n, \kIdx{i}}
				        \\&= \bigoplus_{i=1}^{k} \GRS{N_{\aut}(\b^{q-1}), (\b^{q-1})^{\doubleBrackWIndex{i-1}{\frob}} \star \b; n, 1}.
		\end{align*}
		Similarly, the \ac{GLRS} representation of $\Gab{\b; n, k}$ from \eqref{eq:gab_as_gsrs_glrs} and \cref{thm:glrs-decomp} yield
		\begin{equation}
		    \Gab{\b; n, k}
			= \bigoplus_{i=1}^{m} \GRS{N_{\aut}(\1), \frob^{i-1}(\b); n, \kIdx{i}}
		    = \bigoplus_{i=1}^{k} \GRS{N_{\aut}(\1), \frob^{i-1}(\b); n, 1}.
		\end{equation}
		The resulting subcodes are indeed the same since $b^{q-1} = \frob(b) b^{-1} = \ConjFrob{1}{b}$ applies to any $b \in \Fqm$ and the norm of all conjugates of $1$ equals $1$.
		Further, every $b \in \Fqm$ satisfies
		\begin{equation}
		    (b^{q-1})^{\doubleBrackWIndex{i-1}{\frob}} b = b^{q^{i-1}-1} b = \frob^{i-1}(b).
		\end{equation}
		\qed
	\end{proof}
	
	\paragraph{GRS codes.}
	In the case of \ac{GRS} codes, the automorphism $\aut$ equals the identity map.
	Thus, we are automatically in the case $m = 1$, that is $\Fqm=\Fq$, and our results about \ac{GSRS} and \ac{GLRS} decompositions become trivial: we only obtain the \ac{GRS} code itself.
	More precisely, the application of \cref{thm:gsrs-decomp} and \cref{thm:glrs-decomp} to the \ac{GSRS} and the \ac{GLRS} representations of $\GRS{\b, \veclambda; n, k}$ as given in \eqref{eq:grs_as_gsrs_glrs} both yield
	\begin{equation}
	    \GRS{\b, \veclambda; n, k}
		= \bigoplus_{i=1}^{1} \GRS{N_{\aut}(\b), \b^{\doubleBrackWIndex{i-1}{\id}} \star \veclambda; n, \kIdx{i}}
	\end{equation}
	for distinct code locators $\b \in \Fqm^n$ and nonzero column multipliers $\veclambda \in \Fqm^n$.
	Clearly, $\kIdx{1} = k$ and $b^{\doubleBrackWIndex{i-1}{\id}} = 1$ for all $b \in \Fq$ since $i = 1$ is the only valid option.
	Further, the norm $N_{\aut}(b)$ satisfies $b^{\doubleBrackWIndex{m}{\id}} = b$ for every $b \in \Fq$.
	The combination of these insights yields precisely $\GRS{\b, \veclambda; n, k}$ on the right-hand side of the equation.

	\section{Square code distinguishers}
	\label{sec:square_code_dist}
	
	We now exploit the decomposition of \ac{GSRS} and \ac{GLRS} codes we have seen in the last section to analyze the behavior of the squares of such codes.
	Recall therefore that the \emph{star product} (also called the \emph{Schur product}) of two vectors is their coordinate-wise product, i.e.,
	$$ \x\star \y:= (x_1y_1,\dots, x_n y_n)$$
	for any $\x,\y \in \Fqm^n$.
	For two linear codes $\mycode{C}_1, \mycode{C}_2 \subseteq \Fqm^n$, we define their star product $\mycode{C}_1 \star \mycode{C}_2$ as the code spanned by all pairwise star products $\c_1 \star \c_2$ of codewords $\c_1 \in \mycode{C}_1$ and $\c_2 \in \mycode{C}_2$.
	If $\mycode{C}_1 = \mycode{C}_2= \mycode{C}$, we call $\mycode{C} \star \mycode{C}$ the square (code) of $\mycode{C}$ and denote it by $\mycode{C}^{\star 2}$.
	Let us first state how the square code of random codes behaves.
	
	\begin{lemma}[{Squares of random codes \cite[Thm.~2.2 and Thm. 2.3]{CascudoCramerEtAl2015SquaresRandomLinear}}]
		\label{lem:random_squares}
		Consider a random linear code $\mycode{C} \subseteq \Fqm^n$ of dimension $k$ and let $\delta \defeq \left| \tfrac{k(k+1)}{2} - n \right|$.
		Then, there exists a positive $\gamma \in \RR$ such that for all large enough $k$ the equality
		$$\dim(\mycode{C}^{\star 2}) = \min \left\{\binom{k+1}{2}, n\right\} = \min \left\{\frac{k(k+1)}{2}, n\right\}$$
		is satisfied with probability at least $1 - 2^{\delta \gamma}$.
	\end{lemma}
	
	Using the fact that \ac{GSRS} and \ac{GLRS} codes contain GRS codes as subcodes, we can deduce the following theorem:
	
	\begin{theorem}[Squares of GSRS and GLRS codes]
		\label{thm:squares_of_gsrs_glrs}
		Let $\vecalpha, \veclambda \in \Fqm^{n}$ be \ac{GSRS} parameters as given in \cref{defi:GSRS} and consider the outcome of \cref{lem:p-indep_conditions} such that $\a, \b \in \Fqm^{\n}$ satisfy $\pi(\vecalpha) = \ConjAut{\a}{\b}$.
		Then, the dimensions of the squares
        $\genSkewRS{\vecalpha, \veclambda; n, k}_\aut^{\star 2}$ and $\GLRS{\b, \a, \veclambda; n, k}_\aut^{\star 2}$ are equal and at most
        $$\begin{cases}
            \min\left\{\frac{k(k+1)}{2}, n\right\}& \text{ if } k\leq m\\
            \min\left\{k(m+1) - \frac{m(m+1)}{2}, n\right\} & \text{ otherwise}.
        \end{cases}$$
	\end{theorem}
	
	\begin{proof}
    We first note that the dimensions of the squares of the GSRS and the GLRS code must be equal due to \cref{lem:glrs-to-gsrs} and \cref{lem:gsrs-to-glrs}.
    
    If $k\leq m$, then the GSRS code constitutes the direct sum of $k$ one-dimensional GRS codes with the same evaluation parameters according to \cref{thm:gsrs-decomp} and \cref{rem:k<m}.
    Then the square code consists of $\tfrac{k(k+1)}{2}$ star products of one-dimensional codes.
    We ignore potential nontrivial intersections to arrive at the upper bound $\tfrac{k(k+1)}{2}$ on the square code dimension of the \ac{GSRS} code.
    
        In the case $k> m$, the decomposition of \cref{thm:gsrs-decomp} consists of $m$ GRS subcodes of dimension $k_i\in \bigl\{\lceil \frac{k}{m} \rceil, \lfloor\frac{k}{m} \rfloor \bigr\}$ for $i=1,\dots,m$. The square code decomposes into the sum of $m$ squares of GRS codes of dimensions $k_1,\dots, k_m$, and $\tfrac{m(m-1)}{2}$ Schur products of two different GRS codes with the same evaluation parameters. Each of the former has an expected dimension of $2k_i -1$, respectively, and we expect that each of the latter has dimension $k_i+k_j-1$ where $k_i$ and $k_j$ denote the dimensions of the codes in the star product.
    
        Let us now obtain an upper bound on the dimension of $\genSkewRS{\vecalpha, \veclambda; n, k}_\aut^{\star 2}$ by summing up the square code dimensions of all parts.
        In other words, we assume that all \ac{GRS} squares and Schur products intersect only trivially.
        We denote $r:= k \mod m = k - m \left\lfloor \tfrac{k}{m} \right\rfloor$ with $0\leq r<m$ and express the sum of the dimensions for the square parts as
        \begin{equation}
            \sum_{i=1}^m (2k_i-1) = \sum_{i=1}^{r}  (2\left\lceil \tfrac{k}{m} \right\rceil -1)  +\sum_{i=r+1}^{m} (2\left\lfloor \tfrac{k}{m} \right\rfloor -1)
        \end{equation}
        and the sum of the dimensions for the mixed star products as
        \begin{align}
            \sum_{i=1}^m \sum_{j=i+1}^m (k_i+k_j -1) &= \sum_{i=1}^{r} \sum_{j=i+1}^{r} (2\left\lceil \tfrac{k}{m} \right\rceil -1) + \sum_{i=1}^{r} \sum_{j=r+1}^{m} (2\left\lfloor \tfrac{k}{m} \right\rfloor)
            \\ &\phantom{=~} + \sum_{i=r+1}^{m} \sum_{j=i+1}^{m} (2\left\lfloor \tfrac{k}{m} \right\rfloor -1).
        \end{align}
        Together with the fact $2\left\lceil \tfrac{k}{m} \right\rceil -1 =2\left\lfloor \tfrac{k}{m} \right\rfloor +1$, we can combine the above and obtain the upper bound
        \begin{align}
			&\frac{r(r+1)}{2} (2\left\lfloor \tfrac{k}{m} \right\rfloor +1) + \frac{(m-r)(m-r+1)}{2} (2\left\lfloor \tfrac{k}{m} \right\rfloor -1) +r(m-r) 2\left\lfloor \tfrac{k}{m} \right\rfloor\\
	        &= (2\left\lfloor \tfrac{k}{m} \right\rfloor -1)\cdot \tfrac{1}{2}\bigl(r(r+1)+(m-r)(m-r+1)+2r(m-r)\bigr)\\
			&\phantom{=~} + r(r+1) + r(m-r)\\
	        &= \frac{m(m+1)}{2}\left(2\left\lfloor\tfrac{k}{m}\right\rfloor-1\right) + r\left(m+1\right)\\
	        &= k(m+1) - \tfrac{m(m+1)}{2}
        \end{align}
        on $\dim(\genSkewRS{\vecalpha, \veclambda; n, k}_\aut^{\star 2})$ as long as this is at most $n$.
        Note that, if $k$ divides $m$, we get $r=0$ and the sums with upper limit $r$ vanish.
		\qed
    \end{proof}

	\cref{thm:squares_of_gsrs_glrs} gives hope to distinguish \ac{GSRS} and \ac{GLRS} codes for the case $k > m$.
	Since $k = m + 1$ implies
	\begin{equation}
	    k(m+1)- \frac{m(m+1)}{2}= (m+1)^2- \frac{m(m+1)}{2}= \frac{(m+2)(m+1)}{2}= \binom{k+1}{2},
	\end{equation}
	we restrict to $k > m + 1$ in the following.
	
	McEliece-like cryptosystems need to disguise the secret code in order to hide the algebraic structure which can be used to decrypt messages.
	We study an isometric disguising strategy for \ac{GSRS} and \ac{GLRS} codes and the Hamming metric:
	
	\begin{lemma}[Squares of disguised GSRS and GLRS codes]
		\label{lem:squares_Hamming_disguising}
		Let $\mycode{C}$ be a code obtained from applying a Hamming-metric isometry to a \ac{GSRS} or \ac{GLRS} code with parameters adhering to \cref{defi:GSRS} and \cref{def:glrs_code}, respectively.
		Then, the square $\mycode{C}^{\star 2}$ behaves as described in \cref{thm:squares_of_gsrs_glrs}.
	\end{lemma}
	
	\begin{proof}
		Hamming-metric isometries correspond to multiplication of a generator matrix with a monomial matrix from the right.
		We can decompose any monomial matrix into the product of a diagonal matrix with nonzero entries on the diagonal and a permutation matrix.
		The first part can be combined with the code's column multipliers, resulting in another \ac{GSRS}/\ac{GLRS} code.
		The second part only scrambles the columns which does not change the fact that the code is a \ac{GSRS} or \ac{GLRS} code.
		This implies the statement.
		\qed
	\end{proof}
	
	Note that the above results yield an efficient distinguisher for \ac{GLRS} and \ac{GSRS} codes as long as their square code dimension differs from the expectancy for a random code.
	More precisely, we obtain the following statement:
	
	\begin{lemma}[Naive square code distinguisher]\label{lem:naive_dist}
		Let $\mycode{C} \subseteq \Fqm^{n}$ be a \ac{GSRS} or \ac{GLRS} code as defined in \cref{defi:GSRS} and \cref{def:glrs_code}.
		Assume that its dimension $k$ satisfies $m+1 < k < \tfrac{n}{m+1} + \tfrac{m}{2}$ and disguise $\mycode{C}$ by means of a Hamming-metric isometry.
		Then, the disguised version of $\mycode{C}$ can be distinguished from a random code with the same length and dimension with high probability by means of its square code dimension.
	\end{lemma}
	
	\begin{proof}
		In order to obtain distinguishability, the expected square code dimension for random and for GSRS/GLRS codes needs to deviate by at least one.
		Thus, it is possible to distinguish when the following two inequalities are satisfied:
	    \begin{align}
	        k(m+1) - \tfrac{m(m+1)}{2} < n \quad &\iff \quad k < \tfrac{n}{m+1} + \tfrac{m}{2} ,\\
	        k(m+1) - \tfrac{m(m+1)}{2} < \tfrac{k(k+1)}{2} \quad &\iff \quad k\notin \{m,m+1\}.
	    \end{align}
        As mentioned above, $k$ needs to be larger than $m+1$ to get a different formula for the expected dimension.
		Therefore, the second condition above is obsolete.
		\qed
	\end{proof}
	
	As with classical generalized Reed–Solomon codes, designers of cryptosystems may attempt to choose parameters for \ac{GSRS} or \ac{GLRS} codes such that their square code dimensions are indistinguishable from those of random linear codes. However, when shortening is taken into account, the supposed advantage of the extension parameter $m$ in the square code dimension estimate becomes negligible.
	
	\begin{theorem}[Square code distinguisher]
		\label{thm:distinguisher}
		Let $\mycode{C} \subseteq \Fqm^{n}$ be a \ac{GSRS} or a \ac{GLRS} code with respect to \cref{defi:GSRS} and \cref{def:glrs_code}, respectively.
		Assume that its dimension $k$ satisfies $m + 1 < k < n - \tfrac{1}{2} (m^2 + 3m)$ and disguise $\mycode{C}$ by means of a Hamming-metric isometry.
		Then
        we can distinguish the disguised version of $\mycode{C}$ from a random code by means of its square code dimension with high probability.
	\end{theorem}
	
	\begin{proof}
		We want to shorten $\mycode{C}$ such that we end up in the naive case of \cref{lem:naive_dist}.
		Let us denote the number of shortened coordinates by $s$ and note that the shortened code is a GSRS code of length $n-s$ and dimension $k-s$.
		Hence, we need to find an integer $s$ that satisfies both
		\begin{equation}
		    k - s > m + 1
		    \quad \text{and} \quad
		    k - s < \tfrac{n-s}{m+1} + \tfrac{m}{2}.
		\end{equation}
		While the first inequality is equivalent to $s \leq k - m - 2$, the second can be simplified to
		\begin{equation}
		    s > \tfrac{1}{m}\bigl((k - \tfrac{m}{2})(m + 1) - n\bigr).
		\end{equation}
		A valid $s$ thus exists if
		\begin{equation}
		    k - m - 2 > \tfrac{1}{m}\bigl((k - \tfrac{m}{2})(m + 1) - n\bigr)
		    \quad \iff \quad
		    k < n - \tfrac{1}{2} (m^2 + 3m)
		\end{equation}
		and \cref{lem:naive_dist} then ensures that the square code dimension successfully distinguishes the shortened code from random ones with high probability.
		In fact, $s = k - m - 2$ is always a viable choice if the necessary conditions from above are satisfied.
		\qed
	\end{proof}
	
	\begin{remark}
	    The restriction $k < n - \tfrac{1}{2} (m^2 + 3m)$ in \cref{thm:distinguisher} can also be interpreted as a condition on the redundancy $n - k$ or on the code length $n$.
	    While the latter reads $n > \tfrac{1}{2} (m^2 + 3m) + k$, the naive approach from \cref{lem:naive_dist} implicitly assumes $n > \tfrac{1}{2}(m^2 + 5m) + 2$ since $m + 1 < k < \tfrac{n}{m+1} + \tfrac{m}{2}$ cannot be satisfied otherwise.
	\end{remark}
	
	Observe that the computation of the square code of a $k$-dimensional code in $\Fqm^{n}$ takes $\OCompl{k^2n^2}$ operations \cite[Prop.~5]{CouvreurGaboritEtAl2014DistinguisherBasedAttacks}.
	Since its dimension can be derived via Gaussian elimination, the above distinguisher is clearly polynomial-time.

	\section{Connections to known distinguishers}
	\label{sec:other_distinguishers}
	
	We use this section to recall distinguishers for special classes of skew evaluation codes from the literature to give context for our results.
	We discuss them by code family: first GRS codes, then Gabidulin codes, and finally LRS codes.
	
	\paragraph{GRS codes.}
	The behavior of \ac{GRS} under the Schur product was studied e.g.\ in \cite{Wieschebrink2010CryptanalysisNiederreiterPublic,CouvreurGaboritEtAl2014DistinguisherBasedAttacks}.
	It holds $\dim (\GRS{\b, \veclambda; n, k}^{\star 2}) = \min\{ 2k-1,n\}$ due to the code equality $\GRS{\b, \veclambda; n, k}^{\star 2} = \GRS{\b, \veclambda^{\star 2}; n, 2k - 1}$.
	We obtain the resulting known distinguisher as a special case of \cref{thm:distinguisher}, as the \ac{GRS} setting with $m = 1$ yields $k(m+1)-\tfrac{m(m+1)}{2} = 2k -1$.
	In particular, both distinguishers share the restriction $k \neq 2$ since $k=2$ results in squares of dimension $3$ both for random and \ac{GRS} codes.
	
	Remark that there are efficient key-recovery attacks for \ac{GRS} codes \cite{SidelnikovShestakov1992InsecurityCryptosystemsBased,CouvreurGaboritEtAl2014DistinguisherBasedAttacks} whose analogs for \ac{GSRS} and \ac{GLRS} codes are interesting open research questions.
	
	\paragraph{Gabidulin codes.}
	We first note that our square code distinguisher in the GSRS (or GLRS) representation cannot be applied, since Gabidulin codes require $m\geq n$, which in turn implies that $m\geq k$.
	By \cref{rem:k<m} we thus expect the square code to have the same dimension as the square of a random code.
 
	Many McEliece-like cryptosystems using Gabidulin codes without advanced disguising strategies have been broken with attacks based on the following behavior related to the Frobenius automorphism $\frob$ (see e.g.\ \cite{Overbeck2005NewStructuralAttack,Overbeck2007StructuralAttacksPublic,Horlemann-TrautmannMarshallEtAl2018ExtensionOverbecksAttack,Horlemann-TrautmannMarshallEtAl2016ConsiderationsRankBased}):
	Consider a code $\mycode{C} \subseteq \Fqm^n$ of dimension $k$ over a large enough field $\Fqm$.
	Then, with high probability, the dimension of the vector space $\mycode{C} + \frob(\mycode{C})$ equals $\min\{2k, n\}$ if $\mycode{C}$ is a random code but $\min\{k+1, n\}$ for $\mycode{C}$ being a Gabidulin code.
	In fact, it is true that $\Gab{\b; n, k} + \frob(\Gab{\b; n, k}) = \Gab{\b; n, k+1}$.
	The property that the dimension increases by exactly one makes it easy to iteratively apply it until a code of codimension one is achieved.
	This can then be used to derive (equivalent) code locators $\b$ via the known duality relation for Gabidulin codes.
	
	Observe that this behavior does not carry over to all \ac{GSRS} and \ac{GLRS} codes.
	Nevertheless, we have a look at  how $\mycode{C} + \frob(\mycode{C})$ behaves for a \ac{GSRS} code $\mycode{C}$ as it is interesting to see how the pieces fall into place for Gabidulin codes.
    The $i$-th row of the generator matrix $\skewVandermonde{k}{\vecalpha} \cdot \diag(\veclambda)$ of $\genSkewRS{\vecalpha, \veclambda; n, k}_\aut$ is $\vecalpha^{\doubleBrackAut{i-1}} \cdot \diag(\veclambda)$, and applying $\frob$ to it yields
	\begin{equation}\label{eq:frob_row}
	    \frob(\vecalpha^{\doubleBrackAut{i-1}} \cdot \diag(\veclambda)) = \vecalpha^{q \cdot \doubleBrackAut{i-1}} \cdot \diag(\veclambda^q)
		\quad \text{for } i = 1, \dots, k.
	\end{equation}
	If we want to achieve the analog $\genSkewRS{\vecalpha, \veclambda; n, k}_\aut + \frob(\genSkewRS{\vecalpha, \veclambda; n, k}_\aut) = \genSkewRS{\vecalpha, \veclambda; n, k+1}_\aut$ of the equality that holds for Gabidulin codes, each row in \eqref{eq:frob_row} needs to be an $\Fqm$-linear combination of rows of $\skewVandermonde{k+1}{\vecalpha} \cdot \diag(\veclambda)$.
	In other words, we would need coefficients $\mu_{i, 1}, \dots, \mu_{i, k+1} \in \Fqm$ for every $i = 1, \dots, k$ such that for each column $j = 1, \dots, n$ it holds
	\begin{gather}\label{eq:overbeck-col-mult}
	    \alpha_j^{q \cdot \doubleBrackAut{i-1}} \lambda_j^q
		= \sum_{l=1}^{k+1} \mu_{i, l} \alpha_j^{\doubleBrackAut{l-1}} \lambda_j.
	\end{gather}
	This is only true for very particular column multipliers such as the ones of Gabidulin codes which satisfy $\Gab{\vecalpha; n, k} = \genSkewRS{\vecalpha, \vecalpha^{\frac{1}{q-1}}; n,k}_{\sigma}$ in the setting $\aut = \frob$ according to \eqref{eq:gab_as_gsrs_glrs}.
	
	\paragraph{\ac{LRS} codes.}
	There is an Overbeck-like distinguisher for \ac{LRS} codes which runs in polynomial time if the evaluation parameters, i.e., the vector $\a$ in \cref{def:glrs_code}, are known \cite{HoermannBartzEtAl2023DistinguishingRecoveringGeneralized}.
	The respective operator for a vector $\a \in \Fqm^{\n}$ and any $j \in \NN$ is defined as
	\begin{equation}\label{eq:def_full_operator}
		\Gamma_{\a}^j: \quad \Fqm^{k \times \n} \to \Fqm^{(j+1)k \times \n},
		\qquad \M \mapsto
		\begin{pmatrix}
			\M \\
			\opfull{\a}{\M} \\
			\vdots \\
			\opfullexp{\a}{\M}{j}
		\end{pmatrix}.
	\end{equation}
	While a generator matrix $\G \in \Fqm^{k \times \n}$ of the \ac{LRS} code $\linRS{\b, \a; n, k}_\aut$ yields $\rk(\Gamma_{\a}^j(\G)) = k + j$ for $j =0, \dots, n-k$, a random full-rank matrix $\R \in \Fqm^{k \times \n}$ whose blocks have full column rank over $\Fq$ likely satisfies $\rk(\Gamma_{\a}^j(\R)) = \min\{(j+1)k, n\}$.
	In fact, the code generated by $\Gamma_{\a}^j(\G)$ is precisely $\linRS{\b, \a; n, k+j}_\aut$ for $j = 0, \dots, n-k$.
	The next lemma shows that we can translate this distinguisher from \cite{HoermannBartzEtAl2023DistinguishingRecoveringGeneralized} into the \ac{GSRS} setting and that it behaves conceptually similar.
	While $\Gamma_{\a}^j(\cdot)$ applied to a \ac{LRS} generator matrix in generalized Moore form basically adds $j$ additional rows, the respective \ac{GSRS} counterpart adds $j$ rows to the generator matrix of the corresponding \ac{GSRS} code in skew Vandermonde form multiplied by the diagonal matrix containing the particularly chosen column multipliers.
	
	\begin{lemma}
		Let $\vecalpha \in \Fqm^n$ be a $\PAut$-independent vector with the representation $\pi(\vecalpha) = \ConjAut{\a}{\b}$ from \cref{lem:p-indep_conditions}.
		The Overbeck-like distinguisher $\Gamma_{\a}^j(\cdot)$ with $j=0, \dots, n-k$ for \ac{LRS} codes can be applied to the \ac{GSRS} code $\genSkewRS{\vecalpha, \b; n, k}_\aut$ via the transformations between \ac{GSRS} and \ac{GLRS} codes.
		Then, the resulting code in the \ac{GSRS} setting is $\genSkewRS{\vecalpha, \b; n, k+j}_\aut$.
	\end{lemma}
	
	\begin{proof}
	    Choose $\a, \b \in \Fqm^{\n}$ appropriately for \cref{lem:gsrs-to-glrs} and assume $\pi = \id$ without loss of generality such that
	    \begin{equation}
	        \genSkewRS{\vecalpha, \veclambda; n, k}_\aut = \GLRS{\b, \a, \veclambda \star \b^{-1}; n, k}_\aut.
	    \end{equation}
	    Then, we can apply the Overbeck-like distinguisher $\Gamma_{\a}^j(\cdot)$ to the right-hand side if the column multipliers are trivial, i.e., if $\veclambda = \b$.
	    The output is a generator matrix of $\linRS{\b, \a; n, k+j}_\aut$ and can be translated back into the \ac{GSRS} setting via \cref{lem:glrs-to-gsrs}.
	    We obtain $\genSkewRS{\vecalpha, \b; n, k+j}_\aut$, which concludes the proof.
		\qed
	\end{proof}

	\section{Experimental results}
	\label{sec:experiments}
	
	We ran simulations in SageMath~\cite{stein_sagemath} to verify our results experimentally.
	Since \ac{GSRS} and \ac{GLRS} codes can be transformed into each other (see \cref{subsec:code_connections}), we only simulated the square code distinguisher for \ac{GSRS} codes.
	Further, we did not showcase the shortening mechanism from \cref{thm:distinguisher} which makes the naive distinguisher from \cref{lem:naive_dist} applicable to a wider range of parameters.
	
	We picked several parameter sets $(q, m, n, k)$ and observed 100 iterations for each of them.
	In every run, we chose a code $\genSkewRS{\vecalpha, \veclambda; n, k}_\aut \subseteq \Fqm^n$ with respect to a randomly chosen automorphism $\aut$ with fixed field $\Fq$, with random nonzero block multipliers $\veclambda \in \Fqm^n$, and with random $\PAut$-independent code locators $\vecalpha \in \Fqm^n$.
	Even though \cref{thm:distinguisher} shows that Hamming-isometric disguising does not affect the applicability of the square code distinguisher, we decided to disguise the code nevertheless for illustrative purposes.
	As a consequence, we computed the square code dimension of the code generated by $\skewVandermonde{k}{\vecalpha} \cdot \diag(\veclambda) \cdot \M$ where $\M$ denotes a random monomial matrix, i.e., a Hamming-metric isometry.
	We also determined the dimension of the squares of random codes with the same parameters to show when the square code dimension differs and allows to distinguish.
	
	Our experimental results are showcased in \cref{tab:experiments}.
	Since the right-hand side of the condition $k < \tfrac{n}{m+1} + \tfrac{m}{2}$ from \cref{lem:naive_dist} and thus the range of distinguishable codes grows for fixed $q$ and $m$, we picked the maximum possible code length $n = m(q-1)$ in all cases.
	Recall that the upper bound $n \leq m(q-1)$comes from the need for P-independent code locators for \ac{GSRS} codes and \cref{lem:p-indep_conditions}.
	Note that we included in each case the largest $k$ for which the distinguishing conditions are satisfied and the smallest one for which they are not.
	Additionally, we picked another $k$ for which the codes are distinguishable to also collect some data points that are not directly at the breaking point.
	
	The expected square code dimension for both \ac{GSRS} and random codes coincided precisely with the predicted bounds in all our experiments: squares of \ac{GSRS} codes have dimension $\min\bigl\{k(m+1) - \tfrac{m(m+1)}{2}, n\bigr\}$ (see \cref{thm:squares_of_gsrs_glrs}), and the dimension of squares of random codes is $\min\bigl\{\tfrac{k (k+1)}{2}, n\bigr\}$ (see \cref{lem:random_squares}).
	
	\begin{table}
	    \centering
	    \caption{Experimental results for the square code dimension of \ac{GSRS} and random codes with different parameters.}
	    \label{tab:experiments}
	    \begin{tabular}{cccc|cc}
	        \toprule
			\multicolumn{4}{c|}{\textbf{Code parameters}} & \multicolumn{2}{c}{\textbf{Square code dimension}} \\[5pt]
	        $q$ & $m$ & $n$ & $k$ & GSRS codes & Random codes \\
	        \midrule
			$2^4$ & 4 & 60 & 10 & 40 & \numtimes{55}{100} \\
			$2^4$ & 4 & 60 & 13 & 55 & \numtimes{60}{100} \\
			$2^4$ & 4 & 60 & 14 & 60 & \numtimes{60}{100} \\
	        \midrule
			$2^4$ & 6 & 90 & 12 & 63 & \numtimes{78}{100} \\
			$2^4$ & 6 & 90 & 15 & 84 & \numtimes{90}{100} \\
			$2^4$ & 6 & 90 & 16 & 90 & \numtimes{90}{100} \\
	        \midrule
			$2^6$ & 2 & 126 & 32 & 93 & 126 \\
			$2^6$ & 2 & 126 & 42 & 123 & 126 \\
			$2^6$ & 2 & 126 & 43 & 126 & 126 \\
	        \midrule
			$2^6$ & 4 & 252 & 42 & 200 & 252 \\
			$2^6$ & 4 & 252 & 52 & 250 & 252 \\
			$2^6$ & 4 & 252 & 53 & 252 & 252 \\
	        \midrule
			$2^6$ & 6 & 378 & 44 & 287 & 378 \\
			$2^6$ & 6 & 378 & 56 & 371 & 378 \\
			$2^6$ & 6 & 378 & 57 & 378 & 378 \\
			\bottomrule
	    \end{tabular}
	\end{table}

	\section{Conclusion}
	\label{sec:conclusion}

    In this work, we studied the algebraic structure and distinguishability of generalized skew Reed–Solomon (GSRS) and generalized linearized Reed–Solomon (GLRS) codes in the context of code-based cryptography. We first establish explicit transformations between the GLRS and GSRS settings, showing that structural properties of one family can be transferred to the other. Using these transformations, we prove that the dual of a GSRS code is itself a GSRS code under certain conditions. Some of these results generalize previously known statements, while others are new contributions, providing a rigorous algebraic framework for further analysis of these code families.

We then prove that both GSRS and GLRS codes admit a decomposition into direct sums of classical GRS subcodes. As a consequence, we determine bounds on the dimension of the square code for these codes and construct a polynomial-time distinguisher capable of separating them from random linear codes for large parameter regimes. The distinguisher remains valid under Hamming-metric isometries, and, through shortening operations, applies to a well-defined range of code parameters.

Finally, we emphasize that the existence of a distinguisher does not immediately yield key recovery or a complete cryptographic break of a McEliece-like cryptosystem based on the respective code family. While the distinguisher identifies GSRS and GLRS codes among random linear codes, constructing efficient attacks that recover secret keys from the public code remains an open problem. Addressing this problem requires new techniques beyond the square code criterion, and it constitutes an important direction for future research. Altogether, our results provide precise algebraic insights into GSRS and GLRS codes and establish a framework for assessing their security in post-quantum cryptographic schemes.

	\bibliographystyle{splncs04}
	\bibliography{references}

	\appendix
	\renewcommand{\theHsection}{appendix.\Alph{section}}
	\crefalias{section}{appendix}

	\section{The \reskew cryptosystem}
	\label{sec:reskew}
	
	We use this appendix to introduce the \reskew cryptosystem which is built on \textbf{Re}ed--\textbf{S}olomon codes in a s\textbf{kew} setting and uses Hamming-isometric disguising.
	It allows us to showcase the cryptographic potential of \ac{GSRS} codes by selecting parameter sets adhering to \ac{NIST}'s security levels and providing their key and ciphertext sizes for comparison with state-of-the-art schemes.
	Moreover, \reskew is a straightforward example of a code-based cryptosystem whose public keys are susceptible to the square code distinguisher presented in the main paper.
	As pointed out in \cref{sec:distinguishers_and_crypto}, efficient distinguishing does not directly break a McEliece-like cryptosystem but we understand the skepticism the community usually shows for systems based on distinguishable algebraic codes.
	
	In the design of the scheme, we closely follow Classic McEliece \cite{ClassicMcEliece-spec} which advanced to the fourth and last round of \ac{NIST}'s standardization project for post-quantum \acp{KEM} \cite{AlagicAponEtAl2022StatusReportFourth}.
	In particular, \reskew employs the Niederreiter framework~\cite{Niederreiter1986KnapsackTypeCryptosystems}, which is equivalent to McEliece's original approach to code-based cryptography~\cite{McEliece1978PublicKeyCryptosystem} according to~\cite{LiDengEtAl1994EquivalenceMceliecesNiederreiters}.
	We also adopt systematic public keys, that is, we do not use a whole parity-check matrix $\H \in \Fqm^{(n-k) \times n}$ but only the nonidentity part $\T$ of its systematic form $(\I_{n - k} \mid \T)$.
	
	As in the main part of the paper, we work with the Hamming metric even though many \ac{GSRS} codes have good properties with respect to the skew metric as well (see \cref{lem:skew-metric-gsrs}).
	Our approach has the advantage that we can employ well-established estimators for the complexity of generic decoding when selecting parameters \cite{cryptographicestimators}.
	Nevertheless, an adaptation of \reskew to the skew metric probably allows to reduce the public key size even further.
	It is also possible to translate the system to the sum-rank metric by switching to \ac{GLRS} codes.

	\subsection{System description}
	\label{subsec:reskew-description}

	We present \reskew as a \ac{PKE} scheme and thus describe it by means of the three algorithms for key generation, encryption, and decryption in the following.
	We start with a parameter set of the form $(q, m, s, n, k, t)$ containing the following:
	\begin{itemize}
		\item A prime power $q$ and an integer $m > 1$ which determine the considered field extension $\Fqm$ over $\Fq$.
		\item An integer $0 < s < m$ which chooses the $\Fqm$-automorphism $\aut$ as $\aut(x) = x^{q^s}$ for $x \in \Fqm$ and thus decides on the skew-polynomial ring $\SkewPolyringZeroDer$.
		\item Three integers $n$, $k$, and $t$ which define the code's length $n \leq m (q - 1)$, its dimension $k < n$, and a decodable error weight $t \leq \left\lfloor \tfrac{1}{2}(n - k) \right\rfloor$.
	\end{itemize}
	
	\begin{algorithm}[ht]
		\caption{\reskew key generation.}\label{alg:key_gen}
		
		\Input{\reskew parameter set with $\params = (q, m, s, n, k, t)$.}
		
		\Output{\reskew key pair $(\pk, \sk)$.}
		
		\BlankLine
		
		\Fn{\functionTitle{KeyGen}{$\params$}}{
			Set up the GSRS generator matrix $\Gsec = \skewVandermonde{k}{\b} \cdot \diag(\veclambda) \in \Fqm^{k \times n}$ for random $\PAut$-independent code locators $\b \in \Fqm^{n}$ and random nonzero column multipliers $\veclambda \in \Fqm^{n}$. \;
			Compute the systematic form $(\U \mid \I_{k})$ of $\Gsec$. \;
			Set up the GSRS parity-check matrix $\Hpub = (\I_{n-k} \mid -\U^{\top}) \in \Fqm^{(n - k) \times n}$. \;
			\Return{
				public key $\pk = (\T)$ with $\T = -\U^{\top}$,
				secret key $\sk = (\b, \veclambda)$.
			}
		}
	\end{algorithm}
	
	\paragraph{Key generation.}
	\cref{alg:key_gen} displays the \reskew key generation and starts from a selected parameter set.
	After suitable code locators $\b$ and column multipliers $\veclambda$ are drawn uniformly at random, the generator matrix $\Gsec$ is constructed as the Vandermonde-like matrix $\skewVandermonde{k}{\b} \cdot \diag(\veclambda)$.
	Then, its systematic form is computed to disguise the secret parameters $\b$ and $\veclambda$, and to derive the parity-check matrix $\Hpub$.
	Since the identity part of $\Hpub$ can be restored easily, only $\T$ is used as public key $\pk$.
	Thus, the $k (n - k)$ entries of $\T$ determine \reskew's public key size, which is $k (n-k) \cdot \lceil \log_2(q^m) \rceil$ bits.
	The secret key $\sk$ consists of the two length-$n$ vectors $\b$ and $\veclambda$.
	As a consequence, $2 n$ field elements need to be stored and the secret-key size amounts to $2 n \cdot \lceil \log_2(q^m) \rceil$ bits.
	
	Note that the insecure subclasses of \ac{GRS} and Gabidulin codes can be filtered out easily in the third line of~\cref{alg:key_gen}.
	However, the probability of sampling such a code is negligible and we therefore omit the check.
	Further, notice that \ac{GSRS} codes are \ac{MDS} and thus naturally admit a parity-check matrix in systematic form \cite[Prop.~11.4]{Roth2006IntroductionCodingTheory}.
	This ensures that \cref{alg:key_gen} always succeeds.
	
	\begin{algorithm}[ht]
		\caption{\reskew encryption.}\label{alg:encaps}
		
		\Input{\reskew public key $\pk = (\T)$, message $\m \in \Fqm^{n}$ of weight $t$.}
		
		\Output{Ciphertext $\c$.}
		
		\BlankLine
		
		\Fn{\functionTitle{Encrypt}{$\pk, \m$}}{
			Set up $\Hpub = (\I_{n-k} \mid \T)$. \;
			Compute $\c = \m \Hpub^{\top} \in \Fqm^{n-k}$. \;
			\Return{ciphertext $\c$.}
		}
	\end{algorithm}
	
	\paragraph{Encryption.}
	\cref{alg:encaps} depicts the encryption process whose inputs are a public key $\T$ and a message $\m$ of weight $t$ to be encrypted.
	After the parity-check matrix $\Hpub = (\I_{n-k} \mid \T)$ is set up, $\m$ is encrypted as $\m \Hpub^{\top}$.
	Since the ciphertext has length $n - k$, it can be stored in $(n-k) \cdot \lceil \log_2(q^m) \rceil$ bits.
	
	\begin{algorithm}[ht]
		\caption{\reskew decryption.}\label{alg:decaps}
		
		\Input{\reskew secret key $\sk = (\b, \veclambda)$, ciphertext $\c \in \Fqm^{n - k}$.} %
		
		\Output{Decrypted message $\m$.}
		
		\BlankLine
		
		\Fn{\functionTitle{Decrypt}{$\sk, \c$}}{
			Set up $\Gsec = \skewVandermonde{k}{\b} \cdot \diag(\veclambda)$. \;
			Obtain $\cZero = (\c \mid \0) \in \Fqm^{n}$ by appending $k$ zeros to $\c$. \;
			Use $\Gsec$ to find a codeword $\cHat$ with distance at most $t$ from $\cZero$. \;
			Recover $\m = \cZero - \cHat$. \;
			\Return{decrypted message $\m$.}
		}

	\end{algorithm}
	
	\paragraph{Decryption.}
	\reskew's decryption is displayed in~\cref{alg:decaps} and starts from a ciphertext $\c$ and the corresponding secret key $\sk = (\b, \veclambda)$.
	First, the secret generator matrix $\Gsec$ is constructed as $\skewVandermonde{k}{\b} \cdot \diag(\veclambda)$ and a length-$n$ vector $\cZero$ is obtained by appending $k$ zeros to the ciphertext $\c$.
	Then, $\Gsec$ is used to perform bounded-distance decoding with radius $t$ on the vector $\cZero$, that is, the \ac{GSRS} decoder finds a codeword $\cHat$ whose distance to $\cZero$ is at most $t$.
	The decrypted message $\m$ is then retrieved as the difference $\cZero - \cHat$, which we prove in~\cref{lem:decryption} below.
	The main idea is to exploit the systematic form of the public key to derive a vector $\cZero$ for which the unique codeword $\cHat$ with distance at most $t$ can be determined explicitly.
	The lemma adapts the decryption strategy used in Classic McEliece and described in \cite[Sec.~4.4]{ClassicMcEliece-spec} from binary extension fields to arbitrary finite fields.
	
	\begin{lemma}
		\label{lem:decryption}
	    When $\c = \textsc{Encrypt}(\pk, \m)$ is a faithfully generated \reskew ciphertext and $(\pk, \sk)$ is a valid key pair, then $\textsc{Decrypt}(\sk, \c)$ from~\cref{alg:decaps} correctly recovers the original message $\m$.
	\end{lemma}
	
	\begin{proof}
		First observe that the vector $\cZero = (\c \mid \0) \in \Fqm^{n}$ satisfies
		\begin{equation}
			\cZero \Hpub^{\top} = (\c \mid \0) \cdot (\I_{n - k} \mid \T)^{\top} = \c.
		\end{equation}
		This implies that $\cZero - \m$ is a codeword of the \ac{GSRS} code defined by the parity-check matrix $\Hpub$ because
		\begin{equation}
			(\cZero - \m) \cdot \Hpub^{\top} = \cZero \Hpub^{\top} - \m \Hpub^{\top} = \c - \c = \0
		\end{equation}
		applies.
		Further, $\cZero - \m$ has distance $t$ from $\cZero$, as $\distH(\cZero, \cZero - \m) = \wtH(\m) = t$.
		
		Note that the considered \ac{GSRS} code has length $n$ and dimension $k$, and thus a minimum distance of $n - k + 1$.
		Since $t$ is chosen as $t = \left\lfloor \tfrac{1}{2}(n-k) \right\rfloor$, at most one codeword can lie within a distance of $t$ of a given point.
		Therefore, $\cZero - \m$ is the unique codeword of distance at most $t$ from $\cZero$ and the \ac{GSRS} decoder in~\cref{alg:decaps} recovers it as $\cHat$.
		As a consequence, we can recover $\m$ as $\cZero - \cHat = \cZero - (\cZero - \m) = \m$.
		\qed
	\end{proof}

	Remark that the decryption algorithm relies on unique decoding of \ac{GSRS} codes up to half the Singleton bound.
	This can be achieved efficiently by a Berlekamp--Welch-like approach in cubic complexity~\cite{LiuManganielloEtAl2015ConstructionDecodingGeneralized}, for example.

	\subsection{Parameter sets}
	\label{subsec:param-selection}
	
	In the following, we present \reskew parameter sets for each of the security levels that \ac{NIST} defined for their \ac{PQC} standardization process.
	Security levels one, three, and five refer to any break of the system requiring at least as many computational resources as a brute-force key search on AES-128, AES-192, or AES-256, respectively~\cite[Sec.~4.A.5]{nist2016pqccall}.
	Their security-bit equivalents are estimated as \numprint{143}, \numprint{207}, and \numprint{272} bits in~\cite[Sec.~4.A.5]{nist2016pqccall}.
	We take a conservative path and add a security margin of five bits on top of \ac{NIST}'s bit estimates and thus arrive at target security levels of \numprint{148}, \numprint{212}, and \numprint{277} security bits, respectively.
	
	\paragraph{Information-set decoding.}
	Assuming that we choose \ac{GSRS} codes with parameters for which no efficient structural attacks are known, the selection of suitable parameter sets boils down to estimating the complexity of the best known attacks on the underlying hard problem.
	The fastest known strategy for the computational \ac{SDP} is \acf{ISD}, which Prange coined in 1962~\cite{Prange1962UseInformationSets}.
	The main idea is to introduce a permutation matrix $\P \in \Fqm^{n \times n}$ and alter the equality $\e \H^{\top} = \s$ to obtain $\e \P^{-1} \P \H^{\top} = \s$.
	This step is equivalent to switching to another \ac{SDP} instance with parity-check matrix $\H \P^{\top}$ and weight-$t$ error $\e \P^{-1}$.
	The hope is that now all $t$ nonzero entries of $\e \P^{-1}$ are condensed in the first $n - k$ coordinates because this special case allows to recover $\e$ efficiently by means of linear algebra.
	It can be easily checked if the new \ac{SDP} instance has the desired property and thus the process can be repeated for random permutations until it successfully recovers $\e$.
	The expected number of permutations to try until succeeding is
	\begin{equation}\label{eq:prange}
		\frac{\binom{n}{t}}
	    {\binom{n - k}{t}}
	\end{equation}
	and this is the dominating factor of the computational complexity of Prange's algorithm~\cite[Sec.~4.3]{cryptographicestimators}.
	The literature contains many improvements and variants of \ac{ISD} procedures and most of them allow that a constant predefined share of the error weight is still located in the last $k$ entries of the permuted error.
	This reduces the expected number of trials but increases the cost of each iteration since the enumeration of valid second parts of $\e$ takes longer than assuming it being equal to $\0$.
	Note that a series of improvements in terms of the enumeration are tailored to binary fields and not applicable to $q$-ary fields with $q > 2$~\cite[Sec.~4.3]{cryptographicestimators}.
	Thus, these variants are no threat for \reskew due to its need for a nontrivial field extension.
	
	We employ the \texttt{CryptographicEstimators} library~\cite{cryptographicestimators} to obtain complexity estimates for \ac{ISD} algorithms, which let us assess the security of potential parameter sets and reach good tradeoffs between efficiency and security.
	We aim for conservative estimates and stick to the widely adopted \ac{RAM} model, which neglects memory-access costs.
	We use the library's \ac{SDP} estimator for nonbinary finite fields for which the \ac{ISD} variants by Prange, by Lee--Brickell, and by Stern are implemented.
	Note that this estimator counts additions in the respective finite field but neglects multiplications because the latter are assumed to be implemented via lookup tables~\cite[Sec.~4.3]{cryptographicestimators}.
	
	\paragraph{Parameter constraints.}
	Observe that \reskew parameters need to satisfy certain properties.
	In particular, we need to choose a nontrivial automorphism $\aut$ to exclude the \ac{GRS} case and thus prohibit efficient key-recovery attacks.
	This implies that $s$ with $\aut(x) = x^{q ^s}$ can take values from $1$ to $m - 1$ and therefore $m > 1$ is necessary.
	Additionally, the code locators of a \ac{GSRS} code need to be P-independent which yields the restriction $n \leq m (q - 1)$.
	
	We further fix $t = \left\lfloor \tfrac{1}{2}(n - k) \right\rfloor$ as the maximal Hamming weight for which the chosen nontrivial \ac{GSRS} code of dimension $k < n$ can decode errors uniquely.
	This is reasonable since a larger error weight does not increase key and ciphertext sizes but tends to result in higher \ac{ISD} costs.
	
	\begin{table}
		\centering
		\caption{\reskew parameter sets.}
		\label{tab:reskew-params}
		\begin{tabular}{ll|rrrrrr}
			\toprule
			{parameter set} & NIST level & {$q$}          & {$m$}        & {$s$} & {$n$}          & {$k$}          & {$t$}          \\
			\midrule
			ReSkew-1         & Level 1    & \numprint{233} & \numprint{2} & \numprint{1} & \numprint{427} & \numprint{325} & \numprint{51}  \\
			ReSkew-1-bin     & Level 1    & \numprint{256} & \numprint{2} & \numprint{1} & \numprint{427} & \numprint{325} & \numprint{51}  \\
			\midrule
			ReSkew-3         & Level 3    & \numprint{331} & \numprint{2} & \numprint{1} & \numprint{627} & \numprint{465} & \numprint{81}  \\
			ReSkew-3-bin     & Level 3    & \numprint{512} & \numprint{2} & \numprint{1} & \numprint{626} & \numprint{464} & \numprint{81}  \\
			\midrule
			ReSkew-5         & Level 5    & \numprint{457} & \numprint{2} & \numprint{1} & \numprint{842} & \numprint{624} & \numprint{109} \\
			ReSkew-5-bin     & Level 5    & \numprint{512} & \numprint{2} & \numprint{1} & \numprint{842} & \numprint{624} & \numprint{109} \\
			\bottomrule
		\end{tabular}
	\end{table}
	
	\begin{table}
		\centering
		\caption{Key and ciphertext sizes for the \reskew parameter sets from~\cref{tab:reskew-params}, rounded up to full bytes.}
		\label{tab:reskew-props}
		\begin{tabular}{lr|rrr}
			\toprule
			\multirow{2}{*}{parameter set} & NIST  & public key        & secret key      & ciphertext     \\
			& level & (bytes)           & (bytes)         & (bytes)        \\
			\midrule
			ReSkew-1                        & 1     & \numprint{66300}  & \numprint{1708}  & \numprint{204} \\
			ReSkew-1-bin                    & 1     & \numprint{66300}  & \numprint{1708}  & \numprint{204} \\
			\midrule
			ReSkew-3                        & 3     & \numprint{160077} & \numprint{2665} & \numprint{345} \\
			ReSkew-3-bin                    & 3     & \numprint{169128} & \numprint{2817} & \numprint{365} \\
			\midrule
			ReSkew-5                        & 5     & \numprint{306072} & \numprint{3789} & \numprint{491} \\
			ReSkew-5-bin                    & 5     & \numprint{306072} & \numprint{3789} & \numprint{491} \\
			\bottomrule
		\end{tabular}
	\end{table}
	
	\paragraph{Parameter sets.}
	Since public key size is the major drawback of Classic McEliece, we decided to optimize the \reskew parameter sets for small public keys.
	We iterated over many possibilities with the above discussed parameter restrictions in mind and analyzed the bit-security level by means of the \ac{SDP} estimator from the \texttt{CryptographicEstimators} library.
	Then, we picked the parameter set with the smallest public key for each of the target security levels of \numprint{148}, \numprint{212}, and \numprint{277} bits and named them ReSkew-1, ReSkew-3, and ReSkew-5, respectively.
	In addition, we provide the parameter sets ReSkew-1-bin, ReSkew-3-bin, and ReSkew-5-bin based on binary extension fields to simplify the implementation and to avoid storage overhead.
	
	All proposed \reskew parameter sets are depicted in~\cref{tab:reskew-params} and the resulting system properties are displayed in~\cref{tab:reskew-props}.
	Remark that all sets use extension fields of degree two and thus the only valid choice for a nontrivial automorphism, namely the Frobenius automorphism $\sigma$ defined by $\sigma(x) = x^{q}$.
	The code rates for the three security levels are about \numprint{0.76}, \numprint{0.74}, and \numprint{0.74}, which aligns with the rates between \numprint{0.7} and \numprint{0.8} employed in Classic McEliece.
	The resulting public key sizes are about \numprint{66.3} kB, \numprint{0.17} MB, and \numprint{0.31} MB for the security categories one, three, and five, respectively.
	In particular, this means an improvement of the public key size by at least a factor three compared to Classic McEliece for each security level.

\end{document}